%% file: main.tex
\documentclass{article}
\usepackage[utf8]{inputenc}
\usepackage{amsthm}
\usepackage{amsmath}
\usepackage{xcolor}
\usepackage{xspace}
\usepackage{hyperref}
\usepackage[noadjust]{cite}
\usepackage{mathtools}
\usepackage[normalem]{ulem}

\usepackage{todonotes}
\usepackage[normalem]{ulem}

\renewcommand{\paragraph}[1]{\vspace{0.1in}\noindent{\bf \boldmath #1}}

\author{Michael A. Bender\thanks{Stony Brook University. \href{mailto:bender@cs.stonybrook.edu}{\texttt{bender@cs.stonybrook.edu}}} \and Alex Conway\thanks{VMware Research Group. \href{mailto:aconway@vmware.com}{\texttt{aconway@vmware.com}}} \and Mart\'{\i}n Farach-Colton \thanks{Rutgers University. \href{mailto:martin@farach-colton.com}{\texttt{martin@farach-colton.com}}} \and William Kuszmaul \thanks{Massachusetts Institute of Technology. \href{mailto:kuszmaul@mit.edu}{\texttt{kuszmaul@mit.edu}}} \and Guido Tagliavini \thanks{Rutgers University. \href{mailto:guido.tag@rutgers.edu}{\texttt{guido.tag@rutgers.edu}}}}

\title{Tiny Pointers}
\date{}
\input{packages}

\input{defines}

\begin{document}
\maketitle
\thispagestyle{empty}
\input{abstract}
\newpage 
\pagenumbering{arabic}

\input{intro_draft}
\input{definition}
\input{upper-bound-fixed}

\input{upper-bound-variable}

\input{lower-bound-variable}

\input{applications}
\input{balls-n-bins}
\input{ack}

\bibliographystyle{plain}
\bibliography{bibliography}

\end{document}

%% file: packages.tex
\usepackage[letterpaper, margin=1in]{geometry}
\usepackage{xspace}
\usepackage{enumitem}
\usepackage{cleveref}
\usepackage{tcolorbox}
\usepackage{amsmath,amsthm,amssymb,amsfonts}
\tcbuselibrary{skins,breakable}

\usepackage{algpseudocode,algorithm,algorithmicx}

\algrenewcommand\algorithmicrequire{\textbf{Description:}}
\algrenewcommand\algorithmicensure{\textbf{Postcondition:}}
\algnewcommand{\IIf}[1]{\State\algorithmicif\ #1\ \algorithmicthen}
\algnewcommand{\EndIIf}{\unskip\ \algorithmicend\ \algorithmicif}

\usepackage{listings}
\usepackage{thmtools}
\usepackage{thm-restate}
\usepackage{lipsum}

%% file: defines.tex
\DeclareMathOperator{\E}{\mathbb{E}}

\renewcommand{\epsilon}{\varepsilon}

\newcommand{\defn}[1]{\textbf{\emph{#1}}}
\newcommand{\poly}{\operatorname{poly}}

\newcommand{\prob}[1]{\Pr\left[#1\right]}

\newcommand{\expect}[1]{\E\left[#1\right]}

\newcommand{\single}{\mbox{\textsc{Single}}\xspace}

\newcommand{\leftd}[1]{\mbox{\textsc{Left}}[#1]\xspace}
\newcommand{\iceberg}[1]{\mbox{\textsc{Iceberg}}[#1]\xspace}
\newcommand{\reftable}{dereference table\xspace}

\newcommand{\Reftables}{Dereference tables\xspace}

\newcommand{\refarray}{store\xspace}

\newcommand{\overflowarray}{overflow array\xspace}
\newcommand{\overflowarrays}{overflow arrays\xspace}
\newcommand{\balancetable}{load-balancing table\xspace}

\newcommand\numberthis{\addtocounter{equation}{1}\tag{\theequation}}

\newtheorem{thm}{Theorem}
\newtheorem{lem}{Lemma}

\newtheorem{prp}{Proposition}
\newtheorem{clm}{Claim}
\crefname{clm}{Claim}{Claims}

\theoremstyle{definition}

\crefname{lem}{lemma}{lemmas}
\Crefname{lem}{Lemma}{Lemmas}
\crefname{thm}{theorem}{theorems}
\Crefname{thm}{Theorem}{Theorems}

\newcommand{\punt}[1]{}

\newenvironment{tbox}{\begin{tcolorbox}[
		enlarge top by=5pt,
		enlarge bottom by=5pt,
		 breakable,
		 boxsep=0pt,
                  left=4pt,
                  right=4pt,
                  top=10pt,
                  arc=0pt,
                  boxrule=1pt,toprule=1pt,
                  colback=white
                  ]
	}
{\end{tcolorbox}}

\newcommand{\calB}{\mathcal{B}}

\newcommand{\calL}{\mathcal{L}}
\newcommand{\calO}{\mathcal{O}}

\newcommand{\Allocate}{\mbox{\textsc{Allocate}}}
\newcommand{\Dereference}{\mbox{\textsc{Dereference}}}
\newcommand{\Free}{\mbox{\textsc{Free}}}

\newcommand{\Insert}{\mbox{\textsc{Insert}}}
\newcommand{\Delete}{\mbox{\textsc{Delete}}}
\newcommand{\Query}{\mbox{\textsc{Query}}}

\newcommand{\h}{\mathrm{Bin}}

\newcommand{\relaxedpointers}{retrievers\xspace}

%% file: abstract.tex

\begin{abstract}
This paper introduces a new data-structural object that we call the tiny pointer. In many applications, traditional $\log n $-bit pointers can be replaced with $o (\log n )$-bit tiny pointers at the cost of only a constant-factor time overhead. We develop a comprehensive theory of tiny pointers, and give optimal constructions for both fixed-size tiny pointers (i.e., settings in which all of the tiny pointers must be the same size) and variable-size tiny pointers (i.e., settings in which the average tiny-pointer size must be small, but some tiny pointers can be larger). If a tiny pointer references an element in an array filled to load factor $1 - 1 / k$, then the optimal tiny-pointer size is $\Theta(\log \log \log n + \log k) $ bits in the fixed-size case, and $ \Theta (\log k) $ expected bits in the variable-size case. Our tiny-pointer constructions also require us to revisit several classic problems having to do with balls and bins; these results may be of independent interest.

Using tiny pointers, we revisit five classic data-structure problems. We show that:
\begin{itemize}[noitemsep,nolistsep]
\item A data structure storing $n$ $v$-bit values for $n$ keys with constant-time modifications/queries can be implemented to take space $nv + O(n\log^{(r)} n)$ bits, for any constant $r > 0$, as long as the user stores a tiny pointer of expected size $O(1)$ with each key---here, $\log^{(r)} n$ is the $r$-th iterated logarithm.
\item Any binary search tree can be made succinct with constant-factor time overhead, and can even be made to be within $O(n)$ bits of optimal if we allow for $O(\log^* n)$-time modifications---this holds even for rotation-based trees such as the splay tree and the red-black tree.
\item Any fixed-capacity key-value dictionary can be made stable (i.e., items do not move once inserted) with constant-time overhead and $1 + o(1)$ space overhead. 
\item Any key-value dictionary that requires uniform-size values can be made to support arbitrary-size values with constant-time overhead and with an additional space consumption of $\log^{(r)} n + O(\log j)$ bits per $j$-bit value for an arbitrary constant $r > 0$ of our choice.
\item Given an external-memory array $A$ of size $(1+\varepsilon)n$ containing a dynamic set of up to $n$ key-value pairs, it is possible to maintain an internal-memory stash of size $O(n\log \varepsilon^{-1})$ bits so that the location of any key-value pair in $A$ can be computed in constant time (and with no IOs).  
\end{itemize}
These are all well studied and classic problems, and in each case tiny pointers allow for us to take a natural space-inefficient solution that uses pointers and make it space-efficient for free.

\end{abstract}

%% file: intro_draft.tex

\section{Introduction}
\label{sec:intro}

How many bits does it take to store a pointer? If we know nothing
about the pointer except that it references an element in an array of
size $n$, then there is lower bound of $\log n$ bits. 

For many (and perhaps even most) uses of pointers, however, this information-theoretic lower bound does not apply. As we shall see in this paper, even a small amount of prior information about a pointer (e.g., a node's predecessor in a linked list) can be used to defeat the $\log n$ lower bound. 

This paper introduces a general-purpose tool, which we call the \defn{tiny pointer}, for compressing pointers. 
In settings where pointers are used, tiny pointers can often be used instead to eliminate almost all of the space overhead of pointers. 

\paragraph{What is a tiny pointer?}
Suppose $n$ or more users (i.e., Alice, Bob, etc.) are sharing an array $A$ of size $n$.
A user can request a location in $A$ via a function \textsc{Allocate$()$}, which returns a pointer $p$ to a location that is now reserved exclusively for that user, if there is an available location; the user can later relinquish the memory location by calling a function \textsc{Free$(p)$}. Each user promises only to allocate at most one memory location at a time.\footnote{A user $k$ can request more than one location by creating a unique label $\ell$ for each of their allocations.  In this case, we simply treat the ``user'' for the allocation as the concatenation $k \circ \ell$, so the user $k$ can have multiple allocations without violating the uniqueness requirement.}  For example, if Alice calls \textsc{Allocate$()$} to get a pointer $p$, she must call \textsc{Free$(p)$} before calling \textsc{Allocate$()$} again.

How large do the pointers $p$ need to be? The natural answer is that each pointer uses $\log n$ bits. However, the fact that each pointer has a distinct owner makes it possible to compress the pointers to $o(\log n)$ bits. A critical insight is that the same pointer $p$ can mean different things to different users, via the following scheme. A user $k$ can call \Allocate$(k)$ in order to get a tiny pointer $p$; they can dereference the tiny pointer $p$ by computing a function \Dereference$(k, p)$ whose value depends only on $k$, $p$, and random bits; and they can free a tiny pointer $p$ by calling a function \Free$(k, p)$. 

The reason that tiny pointers are not constrained by the information-theoretic lower bound of $\log n$ bits is that $k$ and $p$ \emph{together} encode the allocated location, rather than $p$ alone. Thus this scheme provides a mechanism for how to use information already available about a pointer (namely, who ``owns'' the pointer) to compress the pointer to size $o(\log n)$ bits.

We refer to the algorithms for the functions \Allocate$(k)$/\Dereference$(k, p)$/\Free$(k, p)$, along with the array $A$ and any associated metadata $M$, as a \defn{dereference table}. We will often refer to the users (i.e., the owners of tiny pointers) as \defn{keys} and to the data stored at the allocated locations pointed at by the tiny pointers as \defn{values}. 
A dereference table that stores $b$-bit values in an array of $nb$ bits (and using $O(n)$ bits of metadata) is said to support \defn{load factor} $1 - \delta$ if the table is capable of storing $(1 - \delta)n$ values at a time. 

An ideal dereference table would simultaneously support a load factor with $\delta = o(1)$, tiny-pointer sizes of $o(\log n)$, and constant-time operations. As we shall discuss shortly, we prove a tradeoff curve between the best achievable load factor $1 - \delta$ and the best achievable the tiny-pointer size $s$. Constructing optimal dereference tables on this tradeoff curve is one of the central questions of this paper.

\paragraph{Using tiny pointers to get tiny data structures.}
In addition to constructing dereference tables with tiny pointers, we show that such dereference tables can be used to obtain improved solutions for a number of classic problems:
\begin{itemize}[noitemsep,nolistsep]
\item A data structure storing $n$ $v$-bit values for $n$ keys with constant-time modifications and queries can be implemented to take space $nv + O(n\log^{(r)} n)$ bits, for any constant $r > 0$, as long as the user stores a tiny pointer of expected size $O(1)$ with each key---here, $\log^{(r)} n$ is the $r$-th iterated logarithm.\footnote{That is, $\log^{(1)} n := \log n$  and $\log^{(i+1)}n := \log\log^{(i)} n$.}
\item Any binary search tree storing $n$ sortable keys in $n$ nodes can be made succinct with constant time overhead, and can be made within $O(n)$ bits of optimal with $O(\log^* n)$-time modifications. This holds even for rotation-based trees such as the splay tree, which is conjectured to be dynamically optimal. 

\item Any fixed-capacity key-value dictionary storing $v$-bit values can be made stable (i.e., items do not move once inserted) with constant time overhead an additive $O(\log v)$-bit space overhead per value.
\item Any key-value dictionary that requires uniform-size values can be made to support arbitrary-size values with constant time overhead and with an additional space consumption of $\log^{(r)} n + O(\log j)$ bits per $j$-bit value, where $r > 0$ is an arbitrary constant.
\item Given an external-memory array $A$ of size $(1+\varepsilon)n$ containing a dynamic set of up to $n$ key-value pairs, it is possible to maintain an internal-memory stash of size $O(n\log \varepsilon^{-1})$ bits so that the location of any key-value pair in $A$ can be computed in constant time (and with no IOs).
\end{itemize}
What unifies these problems is that each is easy to solve space-inefficiently with pointers, and the difficulty in solving them space-efficiently stems from the challenge of eliminating the pointer overhead. 

A theme throughout our uses of tiny pointers is the importance of having access to the full tradeoff curve of optimal tiny-pointer constructions. This is because of the need to balance two types of space overheads: that of storing the tiny pointers themselves, and that of storing the dereference table. The former is determined by tiny-pointer size and the latter is determined by load factor.

\paragraph{This paper.}
In this paper, we first develop a
comprehensive theory of tiny pointers. We consider both \defn{fixed-size
tiny pointers} (where every tiny pointer is of bounded size) and \defn{variable-size tiny pointers} (where every tiny pointer is of bounded expected size).
For both types of tiny pointers we determine the optimal tradeoff curve between load factor and tiny-pointer size in dereference tables.  We then go on to  present the five applications of tiny pointers outlined above.
As an ancillary result, we also reinterpret our tiny-pointer constructions as balls-and-bins results. In doing so, we improve on the known bounds for dynamic load balancing in some important parameter regimes.

\subsection{Results: Constructing Optimal Tiny Pointers}
In Sections~\ref{sec:fixed}, \ref{sec:variable}, and~\ref{sec:lower} we develop tight asymptotic bounds for the best achievable tradeoff curve between tiny-pointer size $s$ and the dereference-table load factor $1 - \delta$. 

\paragraph{Optimal tradeoffs for fixed and variable size tiny pointers.}
For fixed-size tiny pointers, we show that for any load factor $1 - \delta \not\in o(1)$, there is a lower bound of $\Omega(\log\log\log n) $ on the tiny-pointer size $s$. On the other hand, parameterizing by $\delta$, we show that it is possible to achieve a fixed tiny-pointer size $ s = O(\log\log\log n + \log \delta ^ { -1 })$, and we give a lower bound showing that this tradeoff curve is tight. 

We show that the $\log \log \log n$ barrier can be eliminated by instead using variable-size tiny pointers. We prove that for any load factor $ 1 -\delta $, it is possible to achieve average tiny-pointer size $ s = O (1 + \log\delta ^ { -1 }) $, and again we prove that this tradeoff curve is tight for all $\delta $. 

For variable-size tiny pointers, our construction offers a remarkably strong concentration bound on each tiny pointer's size: if the expected size is $k$, then the probability of any given allocation returning a tiny pointer of size greater than $k + j$ for any $j > 0$ is \emph{doubly exponentially small} in $j$.

All of our dereference-table constructions guarantee constant-time operations with high probability, that is, with probability $1 - 1/\poly n$. Thus, tiny pointers can be integrated into data structures while incurring only a constant-factor time overhead.

\paragraph{Relationship to balls and bins.}
In Section~\ref{sec:balls-n-bins}, we reinterpret our tiny-pointer results as balls-to-bins results. Notably, we are able to apply our techniques to the dynamic load-balancing problem, where there are $ n $ bins and up to $m=nh$ balls present at a time: for $h \ge 1$, we give a balls-and-bins scheme with $d + 1$ hash functions that achieves maximum load $h + O(\sqrt{h \log (hd)}) + \frac{\log \log n}{d\log \phi_d}$, which significantly improves the state of the art \cite{DBLP:conf/soda/Woelfel06, Vocking03} when $hd = o(\log n)$. 
 
To understand the relationship between dereference tables and balls-and-bins schemes, think of keys as balls that must be assigned to distinct bins. Each ball $x$ has some probe sequence $h_1(x), h_2 (x), \ldots \in [n]$ of bins where it can be placed. Supporting tiny pointers of size $O(s)$ is equivalent to maintaining a dynamic balls-to-bins assignment such that each ball $x$ is in some bin $h_i(x)$ satisfying $i \le  2 ^{O(s)}$.

What makes this balls-to-bins problem interesting is that the same ball can be inserted, removed, and subsequently reinserted over time. The first time that a ball is inserted, its probe sequence $h_1(x), h_2 (x), \ldots$ is independent of the dereference table's state. But if the ball is removed and then later reinserted, then this is no longer the case: the state of the dereference table has now been affected by (and is partially a function of) the probe sequence. The result is that, in this fully dynamic setting, even the behaviors of very simple balls-to-bins schemes (e.g., random probing \cite{Larson4} or linear probing \cite{KnuthVol3,sandersstability}) have resisted theoretical analysis.\footnote{Work in this setting typically treats linear probing and random probing as techniques for building an open-addressed hash table. In the setting where balls cannot be moved after being placed (or equivalently, where hash-table deletions are implemented with tombstones), the only known bound on either random probing or linear probing is due to Larson \cite{Larson4}, who analyzed random probing with random insertions/deletions.}

A key insight in constructing small tiny pointers is that, by designing the probe sequence of each ``ball'' to have a certain careful structure, we can achieve small probe complexity (and thus small tiny pointers) for an arbitrary sequence of ball insertions and removals. The same techniques are also what allow for us to revisit other related problems such as dynamic load-balancing in bins with unbounded capacities.

\subsection{Results: Five Applications to Data Structures}\label{sec:introapp}

We now describe our five applications of tiny pointers in more detail. The first application is to the classic data-structural problem of storing a dynamic set of values associated with keys. The next three applications are each black-box transformations in which we show how to remove space inefficiency from large classes of data structures. And the final application is a new data structure for a classic problem in external-memory storage.

\paragraph{Overcoming the $\Omega(\log \log n)$-bit lower bound for the cost of data retrieval.}
Our first application revisits the classic \defn{data-retrieval problem} \cite{demaine2005dynamic,alstrup2001optimal, dietzfelbinger2008succinct, static2}, in which a data structure must store a $v$-bit value for each of the $k$-bit keys in some set $S$, and must answer queries that retrieve the value associated with a given key.\footnote{Note that queries are required to be for a key $x \in S$---the data structure is allowed to return an arbitrary value if $x \not\in S$.}  In the static case, where the keys/values are given up front, it is possible to solve the retrieval problem with $O(1)$-time queries using $nv + O(\log n)$ bits of space~\cite{dietzfelbinger2008succinct, static2}; but in the dynamic case where keys/values are inserted/deleted over time, and there are up to $n$ keys/value pairs present at a time (with keys taken from some large polynomial-size universe), it is known that any solution to the retrieval problem must use a lower bound of $nv + \Omega(n \log \log n)$ bits of space, even if super-constant-time operations are allowed \cite{demaine2005dynamic, alstrup2001optimal}. This means that the number of metadata bits per value is $\Omega(\log \log n)$ on average, even if the values are of size $v = o(\log \log n)$. 

We show that, by just slightly modifying the specification of the retrieval problem, we can completely dissolve the $\Omega(\log \log n)$-wasted-bits-per-item lower bound. Suppose, in particular, that whenever the user inserts a key/value pair $(x, y)$, they are given back a small \emph{hint} $h$ that they are responsible for storing. (We will guarantee that the hint has constant expected size.) In the future, when the user wishes to recover the value $y$ for $x$, they present both the key $x$ and the hint $h$ to the retrieval data structure. We call this the \defn{relaxed retrieval problem} and we refer to the hints as \defn{tiny \relaxedpointers}.

The relaxed retrieval problem can also be viewed as a relaxation of the tiny-pointer problem: the tiny retriever $h$ is analogous to a tiny pointer, except that the pair $(x, h)$ does not have to fully encode the position of $y$---instead, the relaxed-retrieval data structure can make use of not just $x$ and $h$, but can also make use of a small auxiliary data structure whose purpose is to help recover $y$. 

Given that we have already stated tight bounds for tiny pointers, it is tempting to assume that the same bounds should hold for tiny \relaxedpointers. We find that this is not so. We show how to construct tiny retrievers of expected size $O(1)$, while supporting queries in constant time (with high probability), and allowing for the following tradeoff curve: using time $\Theta(r)$ for insertions/deletions, the size of the data structure becomes $nv + O(n(1 + \log^{(r)} n))$ bits. So, with constant-time operations, we can achieve size, say, $nv + O(n \log\log\log\log\log n)$, and with $O(\log^* n)$-time operations, we can achieve size $nv + O(n)$. Moreover, in the special case where the value length $v$ is sub-logarithmic, satisfying $v \le \frac{\log n}{\log^{(r)} n}$, the space consumption reduces to $nv + O(n)$ bits, even for constant $r$. 

Remarkably, our construction for tiny \relaxedpointers is \emph{itself a direct application of tiny pointers}---in fact, tiny \relaxedpointers are simply variable-length tiny pointers of $O(1)$ expected size. This is because the ability to construct $O(1)$-length tiny pointers into an array with $\Theta(n)$ entries ends up allowing for us to reduce the relaxed retrieval problem to the dictionary problem, for which highly space-efficient solutions are known \cite{supersuccinct}. 

We remark that the distinction between tiny pointers and tiny \relaxedpointers ends up being significant in several of our applications below. In some cases, tiny \relaxedpointers offer a path to remarkable (and unexpected) space efficiency, while in other cases, the smooth tradeoff curve and pointer-like behavior offered by tiny pointers makes them a better fit. The advantage of tiny \relaxedpointers is that they offer a steep tradeoff between time and space; the advantage of tiny pointers is that they offer indirection-less reference to elements, as well as a flexible tradeoff between different \emph{types} of space consumption (pointer size and load factor). 

\paragraph{Succinct rotation-based binary search trees.}
To describe our second application, we first take a digression into the world of succinct binary trees. Since there are at most $4^n$ ordered binary
trees on $n$ nodes, the pointer structure of a
binary tree can be encoded in $O(n)$ bits. This observation has led to
a great deal of work on optimal (and near-optimal) encodings of binary
trees~\cite{MuRaSt01Succinct,RaRa03Succinct,FaMu11Succinct,NaSa14Succinct,CoNa16Succinct,FrGr03OptimalTrees,Pa08Succincter,DaRaSa17Succinct}. Apart from navigation, state-of-the-art trees also support a wide
variety of query operations (e.g., subtree size~\cite{MuRaSt01Succinct, RaRa03Succinct,FaMu11Succinct,NaSa14Succinct,CoNa16Succinct}, depth~\cite{NaSa14Succinct,CoNa16Succinct}, lowest-common ancestor~\cite{NaSa14Succinct,CoNa16Succinct}, level ancestor~\cite{NaSa14Succinct,CoNa16Succinct}, etc.), while also supporting basic dynamism (e.g., inserting/removing leaves~\cite{MuRaSt01Succinct, RaRa03Succinct,FaMu11Succinct,NaSa14Succinct,CoNa16Succinct}, inserting a node in the middle of an edge~\cite{MuRaSt01Succinct, RaRa03Succinct,FaMu11Succinct,NaSa14Succinct,CoNa16Succinct}, compacting a path of length two~\cite{MuRaSt01Succinct, RaRa03Succinct,FaMu11Succinct,NaSa14Succinct,CoNa16Succinct}, etc.).

One natural form of dynamic operation has proven
elusive, however: the known succinct binary trees do not efficiently support rotations.
The lack of support for rotations is especially important for binary
\emph{search trees}, which store a set of $n$ sortable keys in $n$ nodes. Almost all dynamic balanced binary search trees
(e.g., AVL trees~\cite{AdelsonVelskiiMi62}, red-black trees~\cite{GuibasSe78}, splay trees~\cite{SleatorTa85-splay}, treaps~\cite{AragonSe89,SeidelAr96}, etc.) rely on
rotations when modifying the tree. None of these tree
structures can be encoded with the known succinct-tree
techniques.

We give a randomized black-box approach for transforming dynamic binary search trees into succinct data structures. If there are $ n $ keys in the succinct search tree, each of which is $k$ bits long, then the size of the succinct search tree will be $nk + O( n\log^{(r)} n )$ bits. The transformation induces only a constant-factor time overhead on query operations, and only an $O(r)$-factor time overhead on tree modifications. So, for example, if we set $r = O(\log^* n)$, then edge traversals take time $O(1)$, edge insertions/deletions take time $O(\log^* n)$, and the tree structure is encoded using $O(n)$. In contrast, the previous state of the art \cite{NaSa14Succinct} for implementing rotations in space-efficient binary search trees also encoded the tree structure in $O(n)$ bits (actually, $2n + o(n)$ bits) but required $\omega(\log n)$ time to implement a single rotation.

When $r$ is set to be $O(1)$, the fact that running times are
preserved means that other properties, such as dynamic optimality, are
as well. For example, if the splay tree \cite{SleatorTa85-splay} is dynamically optimal (as the
widely believed Dynamic-Optimality Conjecture~\cite{SleatorTa85-splay} posits), then so is the
succinct splay tree.

\paragraph{Space-efficient stable dictionaries.}
Our third application is a black-box approach for transforming any fixed-capacity key-value dictionary into a  stable dictionary with the same operation set and with only a constant-factor time overhead.
If the original dictionary stores $v$-bit values, then the new stable dictionary also stores $ v $-bit values, and uses $O\left(\log v \right) $ extra bits of space per value than does the original data structure. 

Formally, a key-value dictionary (e.g., a binary search tree, hash table, etc.) is \defn{stable} if whenever a key-value pair is inserted, the position in which the value is stored never changes. (This property is sometimes also referred to as referential integrity \cite{sandersstability} or value stability \cite{iceberg}.)
Stability ensures that users can maintain pointers into a data structure without those pointers becoming invalidated by changes to the data structure \cite{originalstability, sandersstability}. Stability is a strict requirement in many library data structures \cite{cplusplus1, cplusplus2, cplusplus3, cplusplus4, cplusplus5, cplusplus6, cplusplus7, cplusplus8} (and it is a core reason why high-performance languages such as C++ use chained hashing \cite{cplusplus1, cplusplus2}, which is stable, instead of more space-efficient alternatives, such as linear probing~\cite{Knuth63, peterson57} or cuckoo hashing~\cite{Pagh:CuckooHash,FotakisPaSa05,DietzfelbingerWe07}).

Empirical research on achieving stability in space-efficient hash tables dates back to the 1980s~\cite{originalstability, sandersstability} (see also discussion in Knuth's Volume~3~\cite{KnuthVol3}) and the resulting techniques have been built into widely-used hash tables released by Google~\cite{abseil} and Facebook~\cite{F14}. On the theoretical side, if a data structure is storing $k$-bit keys and $v$-bit values, where $k, v = O(\log n)$, it is known how to achieve stability at the cost of an extra $\Theta(\log \log n)$ bits of space per value \cite{demaine2005dynamic}, but it is not known whether $\Omega(\log \log n)$ bits per value are \emph{necessary}.\footnote{Interestingly, there are several specific approaches for which $\Omega(\log \log n)$ bits per value are known to be necessary, for example if stability is achieved via perfect hashing (see Theorem 2 of \cite{demaine2005dynamic}).} Our result shows that it is not---stability can be achieved with $O(\log v)$ extra bits per value. This is especially noteworthy in cases where the value-size $v$ is small\footnote{One especially remarkable consequence is the following: if we wish to store $O(1)$ control bits associated with each key in a data structure, and we wish for the positions of those bits to be stable so that a third party who does not have access to the data structure can still access/modify the control bits, then we can accomplish this with only $O(1)$ extra bits of space per item.}. Our result applies to arbitrary fixed-capacity dictionaries, including, for example, the succinct splay tree constructed above.

\paragraph{Space-efficient dictionaries with variable-size values.}
Our fourth application is a black-box approach for transforming any key-value dictionary (designed to store fixed-size values) into a dictionary that can store different-sized values for different keys. The resulting data structure induces a constant-factor time overhead and offers the following guarantee on space efficiency. Let $\log^{(r)} n $ be the $r$-th iterated logarithm and set $r$ to be a positive constant of our choice. The new data structure incurs an additive space overhead of only $O(\log^{(r)} n + \log |x|)$ bits for each value $x$.  (Interestingly, the iterated logarithm $\log^{(r)} n$ in this application comes from an entirely different source than in our previous applications.)

The ability to store variable-length values also yields a simple solution to the \defn{multi-set problem}, which is the problem of how to design a space-efficient constant-time hash table that stores multi-sets of keys (rather than just sets). The multi-set problem was first posed as an open question by Arbitman et al.~\cite{arbitman2010backyard}, who gave a succinct constant-time hash table capable of storing sets but not multi-sets. A series of subsequent works gave solutions to the multi-set problem, first in the case of random multi-sets~\cite{bercea2019fully}, and then very recently for arbitrary multi-sets~\cite{Bercea2020Dictionary}.
The known solutions come with a drawback, however: the bound on space is the same for duplicate keys as it is for non-duplicate keys. So, if there are $m_i$ copies of some key, then they are permitted to take $m_i$ times as much space as a single copy would, even though, in principle, $m_i - 1$ of the copies could be encoded using an $\log m_i$-bit counter. Our transformation gives a simple alternative solution that avoids this drawback and that can even be applied directly to the original hash table of Arbitman et al.~\cite{Arbitman09Deamortized}: by storing the multiplicity of each key as a (variable-length) value, one can support arbitrary multisets at an additional space cost of only $\log^{(r)} n + \log m_i + O(\log \log m_i)$ bits per key, where $m_i$ is the multiplicity of the key and $r$ is a positive constant of our choice; this is remarkably space efficient considering the fact that $\log m_i$ bits are needed just to store the multiplicity. A nice feature of our solution is that it also applies directly to other dictionaries such as, for example, the succinct splay tree discussed earlier in the section.

\paragraph{An optimal internal-memory stash.}
Our final application of tiny pointers revisits one of the oldest problems in external-memory data structures: the problem of maintaining a small internal-memory stash that allows for one to directly locate where elements reside in a large external-memory array. 

In more detail, the problem can be described as follows~\cite{Larson1}. We are given an (initially blank) external-memory array with $(1 + \epsilon)n$ slots, for some parameters $\epsilon, n$. We must maintain a dynamically changing set $S$ of key-value pairs (where keys are distinct) in the array, such that each time a key-value pair $(x, y)$ is inserted into $S$, the pair $(x, y)$ is assigned some permanent position where it resides in the external-memory array. We must then also maintain a small internal-memory data structure $X$, known as a \defn{stash}, that can be used to recover, for each key $x$, precisely where its key-value pair $(x, y)$ is stored in the external-memory array. A stash enables queries to be performed in a \emph{single} access to external memory.

Work on designing space-efficient and time-efficient stashes dates back to the late 1980s \cite{Larson2, Larson3, Larson1}, and is also closely related to the problem of designing space-efficient page tables in operating systems \cite{IntelManual,AMDManual,BenderBhCo21}. The best-known theoretical results are due to Gonnet and Larson \cite{Larson1}, who give a stash that uses only $O(n \log \epsilon^{-1})$ bits of space. A consequence is that, if $\epsilon = \Theta(1)$, the stash uses only $O(n)$ bits. 

Gonnet and Larson's result comes with several drawbacks, however~\cite{Larson1}. First, the stash only offers provable guarantees in the setting where insertions/deletions to $S$ are random; in the case where $S$ is modified by an arbitrary sequence of insertions/deletions/queries, the problem of designing a space-efficient stash remains open. Second, the internal-memory operations on the stash of \cite{Larson1} are not constant-time in the RAM model (or even constant expected time, when $\epsilon = o(1)$).
 
By combining tiny pointers with modern techniques for constructing space-efficient filters, we show that it is possible to construct a stash of size $O(n \log \epsilon^{-1})$ bits that supports constant-time operations in the RAM model (not just in expectation, but even with high probability) and that supports \emph{arbitrary} sequences of insertions/deletions/queries.

%% file: definition.tex

\section{Preliminaries}

\paragraph{Operations.}
A \reftable with $q$-bit-values is a data structure that supports the following operations:
\begin{itemize}[noitemsep]
\item $\textsc{Create}(n, q)$: The procedure creates a new dereference table, and returns a pointer to an array with $n$ slots, each of size $q$ bits. We call this array the \defn{\refarray}. 

\item $\textsc{Allocate}(k)$: Given a key $k$, the procedure allocates a slot in the \refarray to $k$, and returns a bit string $p$, which we call a \defn{tiny pointer}.

\item $\textsc{Dereference}(k, p)$: Given a key $k$ and a tiny pointer $p$, the procedure returns the index of the slot allocated to $ k $ in the \refarray. If $p$ is not a valid tiny pointer for $k$ (i.e., $p$ was not returned by a call to $\textsc{Allocate}(k)$), then the procedure may return an arbitrary index in the \refarray.

\item $\textsc{Free}(k, p)$: Given a key $k$ and a tiny pointer $p$, the procedure deallocates slot $\textsc{Dereference}(k, p)$ from $k$. The user is only permitted to call this function on pairs $(k, p)$ where $p$ is a valid tiny pointer for $k$ (i.e., $p$ was returned by the most recent call to $\textsc{Allocate}(k)$).
\end{itemize}

We say a key $k$ is \defn{present} if it has been allocated more recently than it has been freed; in this case the tiny pointer $p$ returned by the most recent call to $\textsc{Allocate}(k)$ is said to be $k$'s tiny pointer. The user is only permitted to allocate at most one tiny pointer $p$ to each key $k$. That is, each time that $\textsc{Allocate}(k)$ is called to obtain some tiny pointer $p$, the function $\textsc{Free}(k, p)$ must be called before $\textsc{Allocate}(k)$ can be called again.

We say that slot $i$ in the \refarray is \defn{occupied} if there is a present key $k$ with tiny pointer $p$ such that $\textsc{Dereference}(k, p) = i$, and otherwise we say it is \defn{free}. We typically refer to the parameter $n$ (i.e., the number of slots in the \refarray) as the table's \defn{size} or \defn{capacity}.\

\paragraph{Guarantees.}
\Reftables provide the following guarantees: 
\begin{itemize}[noitemsep]
    \item For any two present keys $k_1 \neq k_2$ with tiny pointers $p_1$ and $p_2$, respectively, $\textsc{Dereference}(k_1, p_1) \neq \textsc{Dereference}(k_2, p_2)$.
    \item $\textsc{Dereference}(k, p)$ only depends on $k$, $p$, random bits, and the parameter $n$.
\end{itemize}

The second property ensures that the act of dereferencing a tiny pointer is similar to the act of dereferencing a standard pointer; in both cases, one does not need to access the data structure being pointed into in order to perform the dereference. This ends up being important for several of our applications later. In particular, it ensures that in external-memory applications, each dereference incurs only a single I/O; and it ensures that in data-structure applications, the locations pointed at by tiny pointers are stable (i.e., once a tiny pointer $p$ is allocated to a key $k$, the location that is being pointed at does not change).

\paragraph{Metadata information.}
The \reftable may store metadata in order to perform updates (allocations and frees) efficiently. Metadata can either be stored as part of the store, or in an auxiliary data structure that is permitted to consume up to $O(n)$ bits. In other words, the dereference table is allowed to use $O(n)$ bits (i.e., $O(1)$ bits of overhead per slot) of metadata for ``free'', without that counting towards the space consumption of the store, but any additional metadata must count towards the space consumption of the store. Note that the dereference table is not allowed to store metadata in any slot of the store that is currently allocated.

\paragraph{Failure probability.}
We will permit allocations to have a small failure probability. That is, each allocation is permitted to fail with probability $1/\poly(n)$,
in which case the allocation simply returns a failure message rather than a tiny pointer. In general, if a random event occurs with probability
$1 - 1/\poly(n)$, we say that it occurs \defn{with high probability (w.h.p.)}.

We remark that, when analyzing dereference tables, we shall always assume that the sequence of allocations, frees, and dereferences are determined by an oblivious adversary (i.e., the sequence is determined ahead of time, rather than adapting to the behavior of the dereference table). One consequence of this is that, if a given allocation fails, the only effect on the operation sequence is that the corresponding call to \textsc{free} is removed. 

\paragraph{Load factor.}
Any implementation of a dereference table must also specify an additional parameter $\delta \in [0, 1]$ dictating how full the table is allowed to be. This means that the dereference table can support up to $(1- \delta)n$ allocations at a time---the quantity $1 - \delta$ is referred to as the table's  \defn{load factor}. If the $\textsc{Allocate}$ function is called when there are already $(1 - \delta)n$ allocations performed, then the dereference table is permitted to fail the allocation.\footnote{Note that, even though a dereference table only guarantees the ability to store up to $(1 - \delta)n$ allocations at a time, we still use the terms ``size'' and ``capacity'' of a dereference table to refer to $n$, rather than $ (1 -\delta) n $, since $n$ represents the total number of $q$-bit entries in the store.}

Since dereference tables can use up to $O(n)$ space for metadata, the total amount of space consumed by a dereference table may be as large as
$nq + O(n) = (1 - \delta)nq + \delta n q + O(n)$. The first term $(1 - \delta)nq$ is space that allocations can make use of, and the other terms $\delta n q + O(n)$ are wasted space. Note that there is no point in considering $\delta \ll 1/q$, since this just makes it so that they wasted space is dominated by metadata. Thus, when constructing a reference table with some load factor $1 - \delta$, we shall always implicitly assume that $q \ge \Omega(\delta^{-1})$. 

\paragraph{Hashing and independence.}
Our dereference-table constructions will all make use of hash functions. For simplicity, we shall treat hash functions in this paper as being uniform and fully independent. This assumption is without loss of generality since there are already known families of hash functions \cite{pagh2008uniform, dietzfelbinger2003almost} that simulate $n$-independence with constant-time evaluation and $O(n)$ random bits, and there are already well understood techniques \cite{liu2020succinct, arbitman2010backyard} for applying these families to data structures that require $\poly n$-independence\footnote{The basic idea is to simply replace the data structure of capacity $n$ with $n^{1 - \epsilon}$ data structures of capacity $n^{\epsilon}$. Each element $x$ in the full data structure gets hashed at random to one of the $n^{1 - \epsilon}$ data structures, each of which only requires $\poly(n^\epsilon) = o(n)$ independence.}. These known techniques can easily be applied directly to all of our data structures; the only caveat is that the families of hash functions being used  \cite{pagh2008uniform, dietzfelbinger2003almost}  introduce their own additional $1/\poly(n)$ failure probability to the data structure. So, even if a data structure offers sub-polynomial failure probability under the assumption of fully random hash functions, if we wish to use an explicit family of hash functions, then we must allow for a $1/\poly(n)$ failure probability.

%% file: upper-bound-fixed.tex

\section{Upper Bound for Fixed-Size Pointers}
\label{sec:fixed}

In this section, we give optimal constructions for fixed-size tiny pointers. We prove the following theorem:
\begin{thm}
\label{thm:upper-bound-fixed}
For every $\delta \in (0, 1)$ there is a \reftable that (i) succeeds on each allocation w.h.p., (ii) has load factor at least $1 - \delta$, (iii) has constant-time updates w.h.p., and (iv) has tiny pointers of size $O(\log \log \log n + \log \delta^{-1})$.
\end{thm}
In particular, for $\delta = 1 / \log \log n$, we get tiny pointers of size $O(\log \log \log n)$. Thus, we can doubly-exponentially beat raw $\log n$-bit pointers, while still supporting a load factor of $1 - o(1)$.

The proof is the simplest of our tiny-pointer constructions, and makes use of two algorithmic building blocks.

\paragraph{The first building block: load-balancing tables.} 
A \balancetable is a simple type of \reftable that has a very specific internal representation, and that, unlike normal dereference tables, is permitted to fail on calls to $\textsc{Allocate}$ with a non-negligible probability. Roughly speaking, if a load-balancing table has load factor $1 - \delta$, then the load-balancing table is permitted to fail on a $\delta$-fraction of allocations.

Load-balancing tables are implemented as follows. If the store is of some size $m$, then we partition it into $m/b$ buckets of size $b = \Theta(\delta^{-2} \log \delta^{-1})$. To allocate a key $k$, we hash $k$ into one of the buckets, using a hash function $h$. If bucket $h(k)$ contains a free slot, then we allocate any free slot $i \in [b]$ within that bucket, and we return $i$ as the tiny pointer. Otherwise, all $b$ slots in the bucket are occupied, and the allocation fails. The function $\textsc{Dereference}(k, p)$ can then be implemented to simply return the $p$-th slot in bin $h(k)$. 

Load-balancing tables will serve as a building block in the dereference tables that we construct. The basic idea is that we can use a load-balancing table to handle all but a $\delta$-fraction of allocations, and the remaining allocations can be handled via some other mechanism. Thus, we will need the following lemma which bounds the total number of failed allocations at any given moment:

\begin{lem}
Consider a load-balancing table with size $m$ and load factor $1 - \delta$. Consider a sequence of allocations and frees such that no more than $(1 - \delta)m$ allocations are made at any given moment. If an allocation fails, and the allocation would have been freed at some time $t$, then we consider the allocation to be \defn{alive} up until that time $t$. At any given moment, the number of allocations that have failed and are still alive is $O(\delta m)$ with probability at least $1 - \exp(- \poly(\delta) m)$.
\label{lem:iceberguse}
\end{lem}

We remark that in all of our applications of Lemma~\ref{lem:iceberguse}, we will have w.l.o.g.\ that $\log \delta^{-1} = o(\log m)$ (since, otherwise, we would have $\log \delta^{-1} = \Omega(\log m)$ and so could just use standard $O(\log m)$-bit pointers). Thus the probability bound offered by the lemma will always be at least $1 - \exp(m^{1 - o(1)}) \ge 1 - 1/\poly(m)$.

We defer the proof of Lemma~\ref{lem:iceberguse} to Section~\ref{ssec:iceberglemma} which establishes a more general version of the lemma. Although the proof is nontrivial, due to interdependencies that form from the same key potentially being allocated/freed/reallocated many times, we do not view it as one of the main technical contributes of this paper. This is because Lemma~\ref{lem:iceberguse} follows easily from a lemma that the current authors established in another recent paper on space-efficient hash tables~\cite{iceberg}. Still, we present an alternative proof in Section~\ref{ssec:iceberglemma} both for completeness, as well as because the proof takes a somewhat different (and more elegant) approach than in our past work, and in order to cover a larger parameter regime.

To conclude our discussion of load-balancing tables, we must describe how to implement allocations and frees in constant time. Here, there are two cases, depending on how $b$ compares to the size $n$ of the dereference table that the load-balancing table is being used within.

If $ b \le \log n$, then we can store a $b$-bit bitmap for each bucket indicating which slots in the bucket are free; and we can use standard bit-manipulation on the bitmap to implement the allocation and free functions in constant time.

We take a different approach if $b \ge \log n$. In this case, we claim that without loss of generality, $q = \omega(\log b)$, where $q$ is the size in bits the elements being stored (we will prove this claim in a moment). This claim means that we can keep track of which slots are free in each bucket of a load-balancing table as follows: we simply store a \defn{free list} in each bucket, that is, a linked list consisting of all the free slots, where each free slot contains a pointer to the next free slot in the list. This is possible since each free slot is $q$ bits and each pointer in the linked list needs only $\log b = o(q)$ bits. The $\log b$-bit base pointers of the $m / b$ linked lists can be stored in an auxiliary metadata array of size $O((m / b) \cdot \log b) \le O(m)$, where $m$ is the size of the load-balancing table. The free lists allow for us to implement the allocation and free functions in constant time.

To prove that this free-list approach works, it remains to show that $q = \omega(\log b)$ without loss of generality. Let $1 - \delta$ be the load factor of the full dereference table (that the load-balancing table is part of) and let $1 - \delta'$ be the load factor of the load-balancing table. Since $b \ge \log n$, we must have $\delta'^{-1} = \tilde{\Omega}(\sqrt{\log n})$. In all of our constructions of dereference tables, if we use a load-balancing table with load factor $1 - \delta'$ satisfying $\delta'^{-1} = \tilde{\Omega}(\sqrt{\log n})$ (or even $\delta'^{-1} = \omega(\log \log n)$), we will always have $\log \delta^{-1} \ge \Omega(\log \delta'^{-1})$. Recall that, if a dereference table has load factor $1 - \delta$, then it is assumed that the dereference table is storing objects of size $q \ge \Omega(\delta^{-1})$ bits. Thus, we have that $q = \omega(\log \delta^{-1}) = \omega(\log \delta'^{-1}) = \omega(\log b)$, as desired.

\paragraph{The second building block: a power-of-two-choices dereference table.}
To compensate for the high failure probability of load-balancing tables, we develop our second building block: a simple dereference table that supports $O(\log \log \log n)$-bit tiny pointers and, unlike a \balancetable, has low failure probability. The downside of this second building block is that it only supports a very small load factor.

\begin{lem}
There exists a $\delta$ satisfying $1 - \delta = \Theta(1 / \log \log n)$, such that there is a \reftable that (i) succeeds on each allocation w.h.p., (ii) has load factor at least $1 - \delta$, (iii) has constant-time updates w.h.p., and (iv) has tiny pointers of size $O(\log \log \log n)$.
\label{lem:twochoices}
\end{lem}
\begin{proof}
We partition the store into buckets of size $b = \Theta(\log \log n)$. When $\textsc{allocate}(k)$ is called, the key $k$ is hashed to two buckets $h_1(k), h_2(k) \in [1, n/b]$. The key $k$ is allocated a slot in whichever of the two buckets contains the most free slots. The tiny pointer $p$ is $1 + \log b = O(\log \log \log n)$ bits, and indicates which slot in the two buckets was allocated.

We can think of the allocations as balls that are inserted into bins using the power-of-two-choices rule \cite{Vocking03, DBLP:conf/soda/Woelfel06}, with the same ball possibly being inserted/deleted/reinserted over time. Since the load factor is $\Theta(1 / \log \log n)$, the expected number of balls in each bin is $O(1)$. In this setting, it is known that, w.h.p., the number of balls in the fullest bin is $O(\log \log n)$ \cite{DBLP:conf/soda/Woelfel06}. Thus allocations succeed w.h.p.

Finally, to implement allocations and frees in constant time, we can just use a bitmap to keep track of which slots in each bucket are free; since each bucket is only $O(\log \log n)$ slots, the bitmaps are each only $O(\log \log n)$ bits, and thus each bitmap fits into a machine word. Using standard bit manipulation, the bitmaps can be used to keep track of which slots are free in constant time per allocation/free (and to find a free slot for a given allocation also in constant time). The bitmaps consume a total of $O(n)$ bits of space.
\end{proof}

\paragraph{Putting the pieces together.}
Of course, power-of-two-choices dereference tables are not very useful on their own, because they only support $o(1)$ load factors. We now show how to combine them with load-balancing tables in order to prove \Cref{thm:upper-bound-fixed}.

\begin{proof}[Proof of \Cref{thm:upper-bound-fixed}]
Since we are willing to have tiny pointers of size $\Theta(\log \log \log n + \log \delta^{-1})$, we can assume without loss of generality that
$\delta = o\left(\frac{1}{\log \log n}\right)$.  

We store a $1 - \delta^2$ fraction of the allocations in a load-balancing table of size $m = (1 - \delta / 2) n$ slots that supports load factor $1 - \delta^2 / c$ for some sufficiently large positive constant $c$; we call this the \defn{primary table}. Allocations that fail in the primary table are stored in a secondary table implemented with Lemma \ref{lem:twochoices} to have size $n' := \delta n / 2$ slots and support load factor $1 - \delta' := \Theta(1 / \log \log n')$. If an allocation fails in the secondary table, or if the load factor of the secondary table ever exceeds $\Theta(1 / \log \log n')$, then the allocation fails in the full dereference table as well. Note that the total size (in terms of slots) of the primary and secondary tables is $n$. See Figure \ref{fig:fixed} for a picture of the layouts of the two tables.

Since both the primary and secondary tables are constant time, so is the full dereference table. Additionally, each allocation can return a tiny pointer that is either in the primary table or in the secondary table (plus 1 bit of information indicating which table it is being pointed into). Since the primary and secondary tables both have tiny pointers of size $O(\log \log \log n + \log \delta^{-1})$, the claim about tiny-pointer size is also proven. 

Our final task is to bound the probability of a given allocation failing. Lemma \ref{lem:iceberguse} tells us that the number of allocations in the secondary table will be a most $\delta^2 n$ at any given moment w.h.p. Since the secondary table has $n'= \Theta(\delta n / 2)$ slots, and since $\delta = o\left(\frac{1}{\log \log n}\right)$, it follows that the number of allocations in the secondary table at any given moment is $o(n' / \log \log n) = o(n' / \log \log n')$ with high probability. We therefore get from Lemma \ref{lem:twochoices} that the allocations in the secondary table each succeed with high probability in $n'$. Without loss of generality, $n' \ge \sqrt{n}$ (since otherwise $\delta \le O(1 / \sqrt{n})$, and we can just use standard $\log n$-bit pointers). Thus the allocations in the secondary table each succeed with high probability in $n$.
\end{proof}

\begin{figure}
    \centering
    \includegraphics[scale = 0.6]{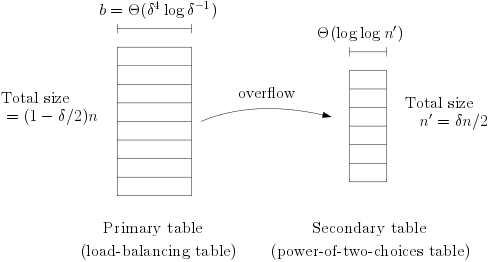}
    \caption{A pictoral representation of the layouts of the primary and secondary tables. The primary table is implemented to support load factor $1 - \Theta(\delta^2)$, so that only $\delta^2 n$ allocations overflow to the secondary table at a time. The secondary table is implemented to have size $n' = \delta n / 2$ and to support a (much sparser) load factor of $\Theta(1 / \log \log n') = \omega(\delta)$, so that it can successfully store all of the overflowed allocations from the primary table.}
    \label{fig:fixed}
\end{figure}

%% file: upper-bound-variable.tex

\section{Upper Bounds for Variable-Sized Pointers}
\label{sec:variable}

In this section, we give optimal constructions for variable-size tiny pointers. We prove the following theorem:
\begin{thm}
\label{thm:upper-bound-variable}
	For every $\delta \in (0, 1)$, there exists a \reftable that (i) succeeds on each allocation w.h.p., (ii) has load factor at least $1 - \delta$, (iii) has constant-time updates w.h.p., and (iv) has tiny pointer size $O(P + \log \delta^{-1})$, where $P$ is a random variable such that $\prob{P \geq i} \leq 2^{-2^{\Omega(i)}}$ for all $i$. In particular, the tiny pointer size is $O(1 + \log \delta^{-1})$ in expectation.
\end{thm}

We can assume without loss of generality that $1 - \delta < \alpha$ for some sufficiently small positive constant $\alpha$ of our choice (if $1 - \delta > \alpha$, we can reset $\delta = 1 - \alpha = \Theta(1)$ without changing the guarantee of the theorem).

Observe that, using the same primary/secondary-table construction as in the proof of Theorem \ref{thm:upper-bound-fixed}, we can immediately reduce to the case where the load factor is a positive constant of our choice. Indeed, suppose that we could implement a dereference table $T$ with load factor $\alpha$ for some positive constant $\alpha > 0$ and average tiny pointer size $O(1)$. Then we can use $T$ as the secondary table in the construction: if the entire dereference table supports load factor $1 - \delta$, then the requirement from the secondary table is that it must be able to support $\delta^2 n$ elements using $\delta n / 2$ slots. So as long as $\delta < \alpha / 2$ (which is without loss of generality) then $T$ suffices.

Thus our task of proving Theorem \ref{thm:upper-bound-variable} reduces to the task of proving the following proposition.

\begin{prp}
\label{prop:upper-bound-variable-big-delta}
There exists a \reftable that (i) succeeds w.h.p.\, (ii) has load factor $\Omega(1)$, (iii) has constant-time updates w.h.p.\ in $n$, and (iv) has tiny pointer size $P$, where $P$ is a random variable satisfying $\prob{P \geq i} \leq 2^{-2^{\Omega(i)}}$ for all $i$.
\end{prp}

For ease of discussion, throughout the rest of the section, we use $ n $ to denote the maximum number of items that can be stored in the dereference table (rather than the number of slots), and we aim to construct a dereference table with $O(n)$ slots.

\begin{figure}
    \centering
    \includegraphics[scale = 0.6]{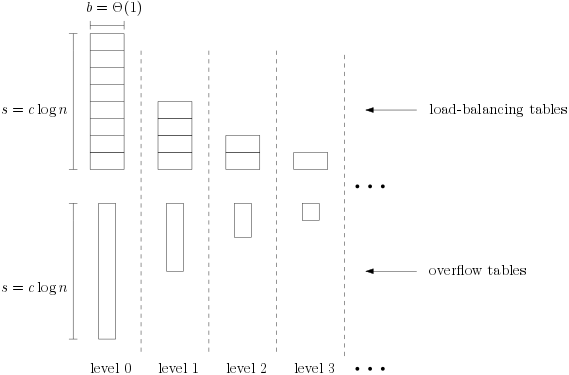}
    \caption{A pictoral representation of the layout used to implement each container of size $\Theta(\log n)$. When an allocation fails in the $i$-th load-balancing table, it either proceeds to the $(i + 1)$-th load-balancing table (if $L_{i + 1} < s_{i + 1}$) or it proceeds to the $i$-th overflow array (which is deterministically guaranteed to have a free slot).}
    \label{fig:variable}
\end{figure}

\paragraph{Constructing the dereference table.}
We now describe our construction for the dereference table that we use to prove \Cref{prop:upper-bound-variable-big-delta}.
The \reftable initially hashes every key into one of $n / \log n$ \defn{containers}, so that, at all times, any container has $\log n$ items in expectation. We deterministically limit the number of elements in each container to $s = c \log n$ items, for some large enough constant $c > 1$ to be determined later. When a key is hashed into a container that already has $c \log n$ items, the allocation fails.

Each container is managed independently, and its allocations/frees are performed using a scheme with $\log_2 s$ levels, as follows. For every $0 \leq i < \log_2 s$, the $i$th level is a \balancetable with $s_i \vcentcolon= s / 2^i$ buckets, each with $b$ slots, for some large enough constant $b \geq 2$ to be determined. 

The basic idea is that, when an allocation in level $i$ fails due to bucket fullness, we recursively attempt the allocation in the next level $i + 1$ (which uses a different hash function than does level $i$). Intuitively, as long as $b$ is a sufficiently large constant, then each level should succeed on at least $1/2$ of its allocations, which is why the next level $i + 1$ can afford to be half the size of the previous one. 

The problem with this basic construction is that if even just a few consecutive levels behave badly, resulting in $\omega(s_{i})$ elements being sent to some level $i$, then there may not be room for those elements in all of the levels $i, \ldots, \log_2 s$ combined. On the other hand, our construction must be able to handle such bad scenarios, because most of the levels are so small that we cannot offer high-probability guarantees on their behavior. Thus, we must modify the construction so that, when a level behaves badly, the effects of that are isolated.

To do this, we add a fallback structure to each level, that we call \defn{\overflowarray}, to prevent excessive occupancy. The \overflowarray in each level $i$ has $s_i$ slots (the same number of slots as the load-balancing table at that level). Let $L_i$ be the random variable denoting the number of values currently stored in levels $i$ or larger, including their \overflowarrays. 
Whenever an allocation at some level $i$ fails (due to bucket fullness), we recursively allocate in the next level only if $L_{i + 1} < s_{i + 1}$, and otherwise, we place the value in any available slot in the \overflowarray of level $i$. The result of this is that we deterministically guarantee $L_i \leq s_i$ for every level $i$ (including level $0$, for which this is trivial, since $s_0 = s$).

Importantly, no \overflowarray can ever run out of space: since $L_i \le s_i$ (deterministically), the total number of elements in the \overflowarray for level $ i $ is also a guaranteed to be a most $s_i$, which is precisely the capacity of the \overflowarray. 

We are now ready to describe the full allocation algorithm. See Figure \ref{fig:variable} for a picture of the layout used to implement each container.

\begin{tbox}
$\textsc{Allocate}(k)$:
\begin{enumerate}
	\item Hash $k$ into one of the $n / \log n$ containers.
	\item If the selected container is already at full capacity $s$, fail. 
	\item Else, allocate $k$ in the selected container:
	\begin{enumerate}
		\item For each $i = 0, 1, \dots, \log_2 s - 1$:
		\begin{enumerate}
			\item Increment $L_i$.
			\item Try to allocate $k$ in the $i$th \balancetable.
			\item If the allocation succeeds:
			\begin{itemize}
			\item Let $j$ be the chosen slot within the chosen bucket. 
			\item Return \texttt{(level $i$, \balancetable, bucket slot $j$)}.
			\end{itemize}
			\item If $L_{i + 1} \ge s_{i + 1}$:
			\begin{itemize}\item Pick any free slot in the $i$-th \overflowarray.
			\item Let $j$ be the chosen slot in the array.
			\item Return \texttt{(level from the back $\log_2 s - 1 - i$, \overflowarray, slot $j$)}.
			\end{itemize}
		\end{enumerate}
	\end{enumerate}
\end{enumerate}
\end{tbox}

Notice that, if an allocation ends up using a slot $j$ in some bucket in the $i$-th level's load-balancing table, then the tiny pointer encodes: the quantity $i$, which is $O(\log i)$ bits; the fact that the allocation used the load-balancing table rather than the overflow array, which is $O(1)$ bits; and the quantity $j$, which is $O(\log b) = O(1)$ bits. The total length of the tiny pointer is $O(\log i)$ in this case.\footnote{We follow the convention that $\log i = \Omega(1)$ for all $i$, so $\log 0$ and $\log 1$ are set to $1$.}

On the other hand, if an allocation ends up using the $j$-th slot in the $i$-th level's overflow array, then the tiny pointer encodes: the quantity $\log_2 s - 1 - i$, which is $O(\log (\log_2 s - 1 - i))$ bits; the fact that the allocation used the overflow array rather than the load-balancing table, which is $O(1)$ bits; and the quantity $j$, which is $O(\log s_i)$ bits. Importantly, in this case, we elect to encode $\log_2 s - 1 - i$, rather than the equivalent quantity $i$. This allows us to bound the total size of the tiny pointer by 
$$O(\log(\log_2 s - 1 - i)) + O(1) + O(\log s_i) = O(\log \log(s / 2^i) + \log s_i) = O(\log \log s_i + \log s_i) = O(\log s_i).$$
Thus, when an allocation uses the overflow array in level $i$, we can bound the tiny-pointer size by $O(\log s_i)$. 

\paragraph{Implementing operations in constant time.}
The information in the tiny pointers allows for dereferences to easily be performed in time $O(1)$. Performing allocations and frees in time $O(1)$ is slightly more difficult, however.

Let us start by considering the na\"ive approach to implementing allocations and see why this is too slow. We must first identify which container to use (this just requires us to evaluate a hash function, taking constant time). We must then determine which level we will be using; if we end up using level $i$, then this takes time $\Theta(i)$, which is too slow when $i = \omega(1)$. 

We solve this problem as follows. Let $d$ to be some sufficiently large positive constant. We will implement levels $0, 1, \ldots, d - 1$ using the naive approach, and then we will implement the levels $d, \ldots, \log_2 s$ using the Method of Four Russians (i.e., the ``lookup-table approach''). Notice that the total number of slots in the levels $d, \ldots, \log_2 s$ is at most $4s_d / 2^d \le (\log n) / 10$. Thus the entire state of which slots are free in those levels can be encoded in $(\log n) / 10$ bits; we store this quantity as metadata for each container, totaling to $O(n)$ bits of metadata across all $n / \log n$ containers. Moreover, the hashes $h_1(k), h_2(k), \ldots, h_{\log_2 s}(k)$ that are used to select a bucket in each level together represent only $O((\log \log n)^2)$ bits (and can be implemented to just be the first $O((\log \log n)^2)$ bits of a single hash function). Thus, the entire state of levels $d, \ldots, \log_2 s$, plus all of the information about the hashes $h_1(k), h_2(k), \ldots, h_{\log_2 s}(k)$, can be encoded in an integer $\phi$ of $(\log n) / 2$ bits that can be constructed in time $O(1)$. This means that we can pre-construct a lookup table of size $2^{(\log n) / 2} = \sqrt{n}$ that we can use to determine, for any given value of $\phi$, which level the allocation should use. The lookup table takes a negligible amount of metadata space, allows for allocations to be performed in time $O(1)$, and can be constructed in time $\tilde{O}(\sqrt{n})$ during the dereference table's creation.

Now that we have specified how to implement allocations, frees are simple to implement, since they just update the metadata to reflect that the slot has been freed (this just flips a single bit in the metadata).

We have now fully specified the construction and implementation of our dereference table. It remains to analyze its properties, namely the probability of failure, the load factor, and the tiny-pointer sizes.

\paragraph{Probability of failure.}
The only way that an allocation can fail is if there is no room in the container that it hashes to, i.e., the container has $c \log n$ elements already. Otherwise, if the container has fewer than $c \log n$ elements, then the allocation is guaranteed to succeed (but, of course, it is not guaranteed to result in a small tiny pointer).

On average, $\log n$ items hash to any particular container, so by a Chernoff bound the maximum size across all containers is at most $c \log n$ w.h.p.\ in $n$ for some positive constant $c$. By the union bound, this holds for all of the $n / \log n$ containers simultaneously, w.h.p.\ in $n$. Thus, if we pick $s = c \log n$ for some large enough constant $c$, at any point in time, all containers will be below capacity w.h.p.\ in $n$.

\paragraph{Load factor.}
Next, we verify that the total number of slots is $O(n)$. The \reftable for each container uses space $O(\sum_i s_i) = O(s_0) = O(s) = O(\log n)$ slots, and there are $n / \log n$ containers. Hence, the total space is $O(n)$, so the load factor is $\Omega(1)$, as desired.

\paragraph{Tiny pointer size.}
To conclude the proof of Proposition \ref{prop:upper-bound-variable-big-delta}, we analyze the tiny pointer size of a given allocation, conditioned on the event that the allocation doesn't fail. The size of the tiny pointer depends on where the key ends up allocated. Specifically, it is:
\begin{itemize}
	\item $O(\log i)$ if the key is allocated in the $i$th \balancetable;
	\item $O(\log s_i)$ if the key is allocated in the $i$th \overflowarray.
\end{itemize}

Fix an arbitrary container to be the one where the allocation takes place, and consider the following events:
\begin{itemize}
	\item $\calB_i$: the key is allocated in the $i$th \balancetable;
	\item $\calO_i$: the key is allocated in the $i$th \overflowarray;
	\item $\calL_i$: $L_i < s_i$.
\end{itemize}
We will condition on two events: (i) that the element picks the container we fixed, and (ii) that the container contains fewer than $c \log n$ elements (i.e., the allocation doesn't fail). We will drop the conditioning notation for clarity. Let $P$ be the size of the output tiny pointer. Then, by the law of conditional expectation,
\[
\expect{P} \leq \sum_i \prob{\calB_i} \cdot O(\log i) + \sum_i \prob{\calO_i} \cdot O(\log s_i). \numberthis \label{eq:expectation}
\]
We bound each term separately. On the one hand,
\begin{align*}
\prob{\calB_i} &\leq \prob{\overline{\calB_0}, \calL_1, \overline{\calB_1}, \dots, \calL_{i - 1}, \overline{\calB_{i - 1}}}\\
&\leq \prob{\overline{\calB_0}} \cdot \prob{\overline{\calB_1} \mid \overline{\calB_0}, \calL_1} \cdots \prob{\overline{\calB_{i - 1}} \mid \overline{\calB_0}, \calL_1, \dots, \overline{\calB_{i - 2}}, \calL_{i - 1}}. \numberthis \label{eq:prob-hash-table-i}
\end{align*}
For every $j$, the load factor of level $j$ is at most $1 / b$, because there are $L_j < s_j$ items, $s_j$ buckets, and each bucket has capacity $b$. This means that at most $1 / b$ of the bins are full, deterministically, so the probability that a full bucket is chosen at most $1 / b$. Hence, every term in \Cref{eq:prob-hash-table-i} is bounded by $1 / b$, and $$\prob{\calB_i} \leq 1 / b^i \leq 1 / 2^i.$$

On the other hand, 
$$\Pr[\calO_i] \le \Pr[\overline{\calL_{ i + 1}}].$$
We can bound the latter probability using Lemma \ref{lem:iceberguse}. By construction, the load-balancing table in level $i$
always has at most $s_i$ allocations made to it (including the failed ones, since $L_i \le s_i$ and $L_i$ counts both the elements in level $i$ and the elements in levels $i + 1, i + 2, \ldots$); moreover, the allocations and 
frees performed on the table are independent of the randomness used in the table. Assuming that the bucket-size $b$
is a sufficiently large constant, it follows that we can apply Lemma \ref{lem:iceberguse} to deduce that, with probability at least
$$1 - \exp(-\poly(b) s_i) = 1 - \exp(-\Omega(s_i)),$$
the number of failed allocations at level $i$ at any given moment is at less than $s_i / 2 = s_{i + 1}$ (and hence $\calL_{i + 1}$ holds).
Thus, we can conclude that
$$\Pr[\calO_i] \le 1 / 2^{\Omega(s_i)}.$$

Putting the pieces together,
\[
\expect{P} = \sum_i \frac{O(\log i)}{2^i} + \sum_i \frac{O(\log s_i)}{2^{\Omega(s_i)}} = O(1).
\]

Notice that these calculations show that a tiny pointer of size $O(\log \ell)$ has probability $2^{-\Omega(\ell)}$, or, equivalently, a tiny pointer of size $O(\ell)$ has probability $2^{-2^{\Omega(\ell)}}$. This suggests that the tiny pointer size decays at a doubly-exponential rate. We prove this next. For any $\ell$,
\begin{align*}
\prob{P \geq \ell} &\leq \sum\limits_{i :\,O(\log i) \geq \ell}\prob{\calB_i} + \sum\limits_{i :\, O(\log s_i) \geq  \ell}\prob{\calO_i}\\
&= \sum\limits_{i \geq 2^{\Omega(\ell)}}\prob{\calB_i} + \sum\limits_{s_i \geq 2^{\Omega(\ell)}}\prob{\calO_i}\\
&= \sum\limits_{i \geq 2^{\Omega(\ell)}} \frac{1}{2^i} + \sum\limits_{s_i \geq 2^{\Omega(\ell)}}\frac{1}{2^{\Omega(s_i)}}.
\end{align*}
Both sums are dominated by their first terms, and are thus $1 / 2^{2^{\Omega(\ell)}}$. Therefore,
\begin{align*}
\prob{P \geq \ell} &\leq \frac{1}{2^{2^{\Omega(\ell)}}},
\end{align*}
which completes the proof of Proposition \ref{prop:upper-bound-variable-big-delta}. As discussed earlier, Proposition \ref{prop:upper-bound-variable-big-delta}, in turn, implies Theorem \ref{thm:upper-bound-variable}.

\paragraph{Bounding sums of tiny-pointer sizes.}
In our applications of tiny pointers, a common way of using variable-size pointers will be to pack $\Theta\left(\frac{\log n}{\log \delta^{-1}}\right)$ of them into $\Theta(\log n)$ bits.
Therefore, we conclude this section by proving a bound of the total number of bits consumed by a set $S$ of $ O (\log n/\log\delta ^ { -1 }) $ tiny pointers.
\begin{prp}
Using the construction in Theorem \ref{thm:upper-bound-variable}, for any set $S$ of $O\left(\frac{\log n}{\log \delta^{-1}}\right)$ tiny pointers, the sum of their sizes will be $O(\log n)$ bits w.h.p.
\label{prop:aggregatesize}
\end{prp}

\begin{proof}
With high probability, all of the allocations for $ S $ succeed. This means that we can ignore the case where allocations fail, so when an allocation fails, we shall treat it as contributing a tiny pointer of size $0 $. 

Let $ K $ be the set of keys corresponding to the tiny pointers in $ S $. The easy case is if every key $ k\in S $ hashes to a different container; in this case, we can analyze each container separately to conclude that each tiny pointer $\textsc{Allocate}(k)$ independently has length $O(\log \delta^{-1} + P_k)$ bits, where $\Pr[P_k > \ell] \le 2^{-2^{\Omega(\ell)}}$. Applying a Chernoff bound for sums of independent geometric random variables, we can conclude that $\sum_{k \in K} P_k \le O(\log n)$ w.h.p., and thus that the total number of bits consumed by $S$ is $O(\log n)$. 

What if some of the keys $ k\in K $ hash to the same container as others $ k'\in K $? Then we can no longer analyze the lengths of the resulting tiny pointers independently. Let $ X $ denote the set of such keys $ k $. Since each tiny pointer is deterministically at most $ O (\log n) $ bits, we can complete the proof by establishing that, with w.h.p., $|X| = O(1)$. 

Let $ k_1, k_2, \ldots$ denote the keys in $ K $, and let $ X_i $ be the indicator random variable for the event that $ k_i $ hashes to the same container as one of $ k_1, k_2,\ldots, k_{ i -1 } $. Then $|X| \le 2 \sum_i X_i$. On the other hand, each $X_i$ independently satisfies $\Pr[X_i] \le (i - 1)/n \le |S| / n \le O(\log n / n)$. Thus $\sum_i X_i$ is a sum of independent indicator random variables with total mean $O(\log^2 n / n)$. Applying a Chernoff bound, we conclude that $\sum_i X_i = O(1)$ w.h.p./, which completes the proof.
\end{proof}

%% file: lower-bound-variable.tex

\section{Lower Bounds}\label{sec:lower}

In this section we prove that the bounds in \Cref{thm:upper-bound-fixed,thm:upper-bound-variable} are tight. We begin by proving a lower bound for variable-size tiny pointers, since it is then used as part of the proof for the fixed size case.
 
What makes the lower bound for variably sized tiny pointer tricky is that any single tiny pointer might be very small.
For example, the dereference table could have a single special slot that corresponds to the tiny pointer $0$ (for every key),
and then if the dereference table ever wanted to make a single tiny pointer small, it could allocate the special slot.
Thus, our proof treats different types of slots differently: for each slot $j$, we define a potential function $\phi(j)$
indicating how ``useful'' that slot is to a random insertion. The idea is that
insertions that use slots $j$ with small potentials $\phi(j)$ must, on average, have relatively large tiny pointers; but 
insertions that use slots $j$ with large potentials $\phi(j)$ must be rare, since only a relatively small fraction of the slots can have
large potentials, and the number of insertions into them can be bounded by the number of deletions out of them.

\begin{thm}
\label{thm:lower-bound-variable}
Consider a universe $\mathcal{U}$ of keys, where $\mathcal{U}$ is assumed to have a sufficiently large polynomial size. If a dereference table supports variable-sized tiny pointers of expected size $s$ and load factor $1- \delta = \Omega(1)$, then $s = \Omega(\log \delta^{-1})$.
\end{thm}

\begin{proof}
  Let $\mathcal{U}$ be a universe of size $n^c$ where $c$ is a
  sufficiently large constant. Let $\delta < 1/4$. Let $T$ be a
  dereference table with $n$ slots and load factor $1 - \delta$ (i.e.,
 it is capable of allocating up to 
  $(1 - \delta)n$ slots to keys from $\mathcal{U}$ at a time). Moreover,
  suppose that $T$ guarantees an expected tiny-pointer length of at
  most $\mu$. Then we wish to show that
  $$\mu \ge \Omega(\log \delta^{-1}).$$

 To simplify our discussion, we shall think of a key $k \in \mathcal{U}$ as
residing in the location that is allocated to it. Thus allocations correspond to insertions, and
frees correspond to deletions.

  Consider a workload in which the table is initialized to contain
  $(1 - \delta)n$ arbitrary elements, and then we alternate between
  insertions and deletions for $n^{c / 2}$ steps. Each insertion
  selects a random element of $\mathcal{U}$ (with high probability in
  $n$, we never insert an element that is already present), and each
  deletion selects a random element out of those present.

  We treat tiny pointers as taking values in $\mathbb{N}$. If the tiny
  pointer takes value $i$, then it uses $\Omega(\log i)$ bits.
  For each element $x \in \mathcal{U}$, let $h_i(x)$ denote the
  position where $x$ would reside in $T$ if $x$ had a tiny pointer
  with value $i$. Set $\ell = \delta^{- 1} / 32$. For each
  position $j \in [n]$ in the table, define the \defn{potential}
  $\phi(j)$ to be
  $$\phi(j) = \frac{|\{u \in \mathcal{U}, i \in [\ell] \mid h_i(u) = j\}|}{|\mathcal{U}|}.$$

  Call an insertion \defn{safe} if the element $x$ that is inserted is
  inserted into one of positions $h_1(x), \ldots, h_{\ell}(x)$. Call
  an insertion \defn{resource efficient} if the element $x$ that is
  inserted is inserted into a position $j$ satisfying $\phi(j) \le \frac{4 \ell}{n}$.

  The probability that a given insertion
  is both safe and resource efficient is at most
  \begin{align*}
    & \sum_{\substack{\text{empty position }j \in [n] \\  \phi(j) \le  \frac{4 \ell}{n}}} \phantom{f} \sum_{i = 1}^{\ell} \Pr_{x \in \mathcal{U}}[h_i(x) = j] \\
    & = \sum_{\substack{\text{empty position }j \in [n] \\  \phi(j) \le  \frac{4 \ell}{n}}} \phantom{f}  \sum_{i = 1}^{\ell} \frac{1}{|\mathcal{U}|} \sum_{x \in \mathcal{U}} \mathbb{I}_{h_i(x) = j} \\
    & =  \sum_{\substack{\text{empty position }j \in [n] \\  \phi(j) \le  \frac{4 \ell}{n}}} \phantom{f}  \phi(j) \\
    & \le \sum_{\text{empty position }j \in [n]}   \phantom{f}  \frac{4 \ell}{n} \\
    & = \delta n  \frac{4 \ell}{n} \\
    & = \frac{1}{8}.
  \end{align*}
  It follows that the expected number of insertions that are safe and
  resource efficient is at most $n^{c / 2} / 8$.

  Next we bound the expected number of insertions $A$ that are safe
  but not resource efficient. Rather than bound $A$ directly, we
  instead examine the number of \emph{deletions} $B$ where the deleted
  element is deleted from a position $j$ satisfying
  $\phi(j) > \frac{4 \ell}{n}$. Note, in particular, that
  $$A \le B + n.$$

  By the definition of $\phi(j)$, we have that
  $\sum_{j = 1}^n \phi(j) = \ell$. It follows that
  $|\{j \in [n] \mid \phi(j) >  \frac{4 \ell}{n}\}| \le n / 8$. Each random deletion
  therefore has probability at most
  $\frac{n / 8}{(1 - \delta)n} \le 1/4$ of removing an element in a
  position $j$ satisfying $\phi(j) >  \frac{4 \ell}{n}$. Thus $\E[B] \le n^{c / 2} / 4$ which means that
  $$\E[A] \le n^{c / 2} / 4 + n \le n^{c / 2} / 2.$$

  Since the expected number of insertions that are safe and resource
  efficient is at most $n^{c / 2} / 8$ and the expected number of
  insertions that are safe and resource inefficient is at most
  $n^{c / 2} / 2$, the expected number of insertions that are safe is
  at most $\frac{5}{8}n^{c / 2}$. The expected number of insertions
  that are \emph{not} safe is therefore at least
  $\frac{3}{8} n^{c / 2}$. Each unsafe insertion results in a tiny
  pointer of length at least
  $\Omega(\log \ell) = \Omega(\log \delta^{-1})$ bits. Since a
  constant fraction of the insertions are expected to result in a tiny
  pointer of length at least $\Omega(\log \delta^{-1})$, we must have
  $\mu \ge \Omega(\log \delta^{-1})$. 
\end{proof}

Next we prove a lower bound for fixed-sized tiny pointers, which shows that the bound in \Cref{thm:upper-bound-fixed} is tight.
\begin{thm}
\label{thm:lower-bound-fixed}
Consider a universe $\mathcal{U}$ of keys, where $\mathcal{U}$ is assumed to have a sufficiently large polynomial size. If a dereference table supports fixed-sized tiny pointers of size $s$ and load factor $1- \delta = \Omega(1)$, then $s = \Omega(\log \log \log n + \log \delta^{-1})$.
\end{thm}

It suffices to prove that $s= \Omega(\log \log \log n)$, since we have already shown that $s= \Omega(\log \delta^{-1})$.

The proof re-purposes a classic balls-and-bins lower bound.
Say that a ball-placement
rule is  \defn{sequential} if balls are placed sequentially, without
knowledge of future ball arrivals,  and if balls are never
moved after being placed.

\begin{thm}[Theorem 2 in \cite{Vocking03}]
Suppose that $m$ balls are placed sequentially into $m$ bins using an arbitrary sequential ball placement rule choosing $d$ bins for each ball at random according to an arbitrary probability distribution on $[m]^d$. Then the number of balls in the fullest bin is $\Omega((\log \log m) / d)$ w.h.p.
\label{thm:vocking}
\end{thm}

We now prove Theorem \ref{thm:lower-bound-fixed}.

\begin{proof}[Proof of \Cref{thm:lower-bound-fixed}]
Assume for contradiction that there exists a dereference table with load factor $1 - \delta = \Omega(1)$ and that supports fixed-size tiny pointers of size $s = o(\log \log \log n)$ bits. Let $n$ be the number of slots in the dereference table, and let $ m = (1 - \delta)n$ be the maximum number of allocations that the dereference table can support at a time; assume without loss of generality that $1/(1 - \delta) \in \mathbb{N}$, so $n$ is a multiple of $m$. Finally, let $S = 2^s$, and observe that, by assumption, $S = o(\log \log n)$---and since $m = \Theta(n)$,  $S = o(\log \log m)$.

Recall that $\mathcal{U}$ is the universe from which the keys are taken. For each key $x \in \mathcal{U}$, define the sequence $h_1(x), h_2(x), \ldots, h_{S}(x) \in [m]$ so that $h_i(x) = \lfloor \frac{m}{n} \Dereference(x, i) \rfloor$. Note that, by the definition of the $\Dereference$ function, the sequence $h_1(x), h_2(x), \ldots, h_{S}(x)$ is a function of only on $x$, $i$, $n$, and the random bits of the dereference table---therefore, the sequence is predetermined by the coin flips, and is independent of the sequence of allocations/deallocations that are performed. Let $R \in [m]^{S}$ be a random variable obtained by selecting $x \in \mathcal{U}$ at random and setting $R = \langle h_1(x), h_2(x), \ldots, h_{S}(x)  \rangle$; and let $\mathcal{R}$ be the probability distribution for $R$.

We will now construct a sequential ball-placement rule for mapping $m$ balls to $m$ bins. Our rule independently assigns each ball a random bin sequence $\langle h_1, h_2, \ldots, h_{S}\rangle \sim \mathcal{R}$ of $S$ bins. Equivalently, we can think of the $m$ balls as being $m$ keys  $x_1, x_2, \ldots, x_m$, where each $x_i$ is selected uniformly and independently at random from $\mathcal{U}$, and each $x_i$ has a bin sequence of $\langle h_1(x), h_2(x), \ldots, h_{S}(x)  \rangle \in [m]^{S}$. 

Since $|\mathcal{U}| \ge \poly(n)$, we have that with high probability in $n$, the $x_i$'s are distinct. Our ball placement rule uses our dereference table to decide where to place balls. To place ball $x_i$ into a bin, we compute $p_i = \Allocate(x_i)$, and we place $x_i$ into the $p_i$-th bin in $x_i$'s bin sequence, which is given by bin 
$$h_{p_i}(x_i) = \left\lfloor \frac{m}{n} \Dereference(x_i, p_i) \right\rfloor \in [m].$$

In summary, we have constructed a sequential ball placement rule that places $m$ balls sequentially into $m$ bins and that chooses a set of $d = S$ bins for each ball according to a probability distribution $\mathcal{R}$ over $[m]^d$. By Theorem \ref{thm:vocking}, we can deduce that the fullest bin contains at least $$\Omega\left((\log \log m) / d\right) = \Omega\left((\log \log m) / S\right)  = \omega(1)$$
balls with high probability in $m$. 

On the other hand, the dereference table guarantees that $\Dereference(x_i, p_i) \in [n]$ is unique for each $i \in [m]$. The number of balls $x_i$ satisfying 
$$\left\lfloor \frac{m}{n} \Dereference(x_i, p_i)\right\rfloor = j$$
for a given $j$ is therefore at most $\frac{n}{m} = O(1)$. This means that the number of balls in any given bin is also $O(1)$. Since the dereference table succeeds with high probability in $n$, we can deduce that there are $O(1)$ balls in the fullest bin with high probability in $n$. This contradicts the fact that the number of balls in the fullest bin is $\omega(1)$, thereby completing the proof by contradiction.\end{proof}

%% file: applications.tex

\section{Applying Tiny Pointers to Five Problems in Data Structures}
\label{sec:applications}

In this section we present several applications of tiny pointers to classical problems in data structures: 
\begin{itemize}
    \item Relaxed Retrieval: we show that a slight modification to the classic retrieval problem eliminates the classical lower bound of $\Omega(\log \log n)$ wasted-bits-per-item (\Cref{retrieval}).  
    \item Succinct binary search trees: we give an approach for transforming arbitrary dynamic binary search trees into succinct data structures (\Cref{trees}).
    \item Space-efficient stable dictionaries: we transform any fixed-capacity key-value dictionary into a key-value stable dictionary (\Cref{stable}).
    \item Space-efficient dictionaries: we transform any dictionary with fixed-size values into one which can space-efficiently store variably sized values (\Cref{var}).
    \item An optimal internal-memory stash: we construct a constant-time stash that space-efficiently stores the locations of elements residing in a large external-memory data structure (\Cref{stash}).
\end{itemize}

\subsection{Some General-Purpose Techniques for Using Tiny Pointers}\label{applicationprelim}

Before diving into specific applications, we briefly discuss several preliminary definitions and techniques that will be useful in multiple of the applications.

\paragraph{Key-value dictionaries.}
Several of our applications will perform black-box transformations in order to add new features (namely, stability and variable-sized values) to key-value dictionaries. Formally, a \defn{key-value dictionary} (often just called a \defn{dictionary}) is any data structure that stores key-value pairs (e.g., a hash table or a tree), where each key appears at most once. Typically, a key-value dictionary supports insertions, deletions, and queries, where queries, in particular, return the value associated to some key. Depending on the data structure, additional operations may also be supported, for example successor queries, which return the successor to some key. 

We say that a key-value dictionary uses a \defn{value array} if it designates some contiguous chunk of memory (that can be extended or shrunk over time) whose purpose is to store the values corresponding to keys. If values are $ k $ bits long, then the value array can be viewed as a array of $k$-bit objects. 

In our applications, we will restrict ourselves to dictionaries that store their values in value arrays. For simplicity, we will assume that the dictionary uses a single value array, although all of our results can also easily be applied to a dictionary that makes use of many separately-allocated value arrays (as long as each individual value array is at least $\Omega(\log n)$ bits). The reason that we assume a single value array is because, to the best of our knowledge, all of the known space-efficient key-value dictionaries can easily be implemented in this format, so we choose to avoid introducing unnecessary complication to the results.

\paragraph{How to store value arrays of tiny pointers.}
A theme in several of our applications will be to modify a value array so that, rather than storing values directly, we instead store tiny pointers of some size $k$. Recall, however, that tiny pointers of size $k = o(\log \log \log n)$ bits are not fixed-size, meaning that some tiny pointers may require more than $ k $ bits. Nonetheless, if we are willing to use a value-array that is a constant-factor larger, then there is a simple trick, which we call \defn{chunked pointer storage}, that lets us interact with these variable-length tiny pointers in the same way that we would interact with fixed-length tiny pointers. 

Break the value array into contiguous chunks of $O(\log n / k)$ tiny pointers. By Proposition \ref{prop:aggregatesize}, the total number of bits used by the tiny pointers in each chunk is $ O (\log n) $ with high probability in $ n $. Thus each chunk can be stored in $ O (\log n) $ bits, meaning that the entire value array can be stored in $O (n k)$ bits. 

There is, however, the remaining issue of how to efficiently access and modify the $ j $-th tiny pointer in a given chunk. For each chunk, we can store an additional $ O (\log n) $-bit bitmap where the bits that are set to $ 1 $ indicate the positions in the chunk where tiny pointers begin and end. To efficiently find the $ j $-th tiny pointer, it suffices to find the $j$-th and $j + 1$-th $1$s in the bitmap. (The tiny pointer can then be extracted, modified, and reinserted, in constant time using standard bit manipulation on the bitmap and the chunk.) The problem of finding the $j$-th $1$ in a $O(\log n)$-bit bitmap is easily solved with the method of four Russians \cite{ArlazarovDiYeKr1970}: simply store an auxiliary lookup table of size $\sqrt{n}$ that allows for such queries to be answered in a $(\log n) /2$-bit bitmap in a single lookup, and then perform $O(1)$ lookups to perform such a query in an $O(\log n)$-bit bitmap.

\paragraph{How to dynamically resize a data structure using tiny pointers.}
Several of our applications will also encounter the problem of using tiny pointers in a data structure whose size dynamically changes over time. Of course, this means that we must also dynamically resize dereference tables. Our applications will take the following approach, which we call  \defn{zone-aggregated resizing}.

Consider a value array storing tiny pointers to $k$-bit items in a dereference table (and assume $k$ bits fit in $O(1)$ machine words). Suppose that we wish to maintain the dereference table at a load factor of $ 1 -\Theta (1/k) $, that way the number of bits wasted per item stored is $ O (1) $; note that this means that the tiny pointers in the value array are $\Theta(\log k)$ bits on average. Further suppose, however, that the value array dynamically changes size over time (meaning that elements must be added and removed from the dereference table). For our discussion here, we will assume that the value array itself is dynamically resized to always be at a load factor of at least $\Omega(1)$.

How can we update the dereference table to maintain a load factor of $ 1 -\Theta (1/k) $ while the number of items changes over time? Rather than just using a single dereference table, we use $k $ dereference tables, and add $\Theta (\log k) $ bits to each tiny pointer in order to indicate which dereference table is being pointed into (this doesn't change the asymptotic size of the tiny pointers). We can grow and shrink the capacity (i.e., number of slots) of the dereference tables by either (a) rebuilding the smallest dereference table to double its size, or (b) rebuilding the largest dereference table to halve its size. If we assume for the moment that rebuilding a dereference table takes time proportional to the table's size, then the rebuilds can be de-amortized to take time $O(1 )$ per operation (i.e., per modification to the dereference tables), while maintaining the desired load factor of $ 1 -\Theta (1/k) $.

The problem with rebuilding a dereference table is that all of the tiny pointers into that dereference table become invalidated. The actual construction of the new dereference table can easily be performed in linear time, but how do we update the tiny pointers in the value array? If the value array has size $ n $, then the dereference table being rebuilt consists of only $\Theta (n/k) $ items. We want to identify where the tiny pointers to those items are in the value array in time $\Theta (n/k) $ rather than time $\Theta (n) $. 

The solution to this issue is very simple: break the value array into contiguous \defn{zones} each of which consists of $ k $ values. Within each zone, maintain $k$ linked lists, where the $i$-th linked list contains the tiny pointers that point into the $ i $-th dereference table. Importantly, because these linked lists are within a zone of size $k$, the pointers \emph{within} each linked list only require $\Theta (\log k) $ bits each; thus the linked lists do not asymptotically increase the size of the value array. On the other hand, in order to find all of the tiny pointers for a given dereference table, one can simply look at one linked list in each of the $\Theta (n/k) $ zones, allowing for all $\Theta (n/k) $ of the tiny pointers to be identified in time $\Theta (n/k) $.

For reasons that we shall see later, one of our applications will also require us to use larger zones of size $\poly (k) $ rather than just of size $ k $. For now, we simply remark that using larger zones of size $\poly (k) $ still allows for the linked-list overhead of each tiny pointer to be bounded by $\Theta (\log k) $ bits, and that the time needed to identify the tiny pointers to a dereference table of size $j$ is only
\begin{equation}
    O(j + n / \poly(k)),
    \label{eq:rebuild}
\end{equation}
since the number of linked lists that must be examined is only $O(n / \poly (k)) $.

\subsection{Overcoming the $\Omega(\log \log n)$-Bit Lower Bound for Data Retrieval}\label{retrieval}

Our first application revisits the classic retrieval problem \cite{demaine2005dynamic,alstrup2001optimal, dietzfelbinger2008succinct,static2}, in which a data structure must store a $v$-bit value for each of the $k$-bit keys in some set $S$, and must answer queries that retrieve the value associated with a given key. Here, we address the dynamic version of the problem, where the data structure must support the functions $\Insert(x, y)$ (which inserts a new $x \in [2^k]$ into $S$ and associates it with value $y \in [2^v]$), $\Delete(x)$ (which removes some $x \in S$ from $S$), and $\Query(x)$ (which returns the value $y$ corresponding to $x$ for some $x \in S$), allowing for the set $S$ to grow up to some maximum size $n$. Note that, in the retrieval problem, it is the \emph{user's responsibility} to ensure that every invocation of $\Insert$ is on a key $x \not\in S$, every invocation of $\Query$ is on a key $x \in S$, and every invocation of $\Delete$ is on a key $x \in S$. 

It is known that, if $k = (1 + \Omega(1))\log n$ bits, then any solution to the dynamic retrieval problem must use at least $nv + \Omega(n \log \log n)$ bits of space \cite{alstrup2001optimal}, regardless of the time complexity, and even if $v = 1$. It is further known that, if $k = \Theta(\log n)$ and $v = O(\log n)$, then the $nv + \Theta(n \log \log n)$ space bound can be accomplished by a randomized constant-time data structure \cite{demaine2005dynamic}. 

We will now show that, by slightly relaxing the retrieval problem, we can use tiny pointers to obtain significantly better space bounds. In the \defn{relaxed retrieval problem}, the insertion/deletion/query operations are modified to work as follows. The operation $\Insert(x, y)$ now returns a \defn{tiny retriever} $r$ which the user must remember. In the future, if the user wishes to query $x$ (and they have not yet deleted $x$), they call $\Query(x, r)$ to obtain the value $y$. Finally, if the user ever wishes to remove $x$ from the set $S$, then the user calls $\Delete(x, r)$. 

The role of the tiny retriever is similar to that of a tiny pointer---it acts as a hint to assist the data structure. Unlike for tiny pointers, however, the pair $(x, r)$ does not have to fully encode the position of $y$; instead, query operations $\Query(x, r)$ can use auxiliary metadata, beyond just $x$ and $r$, to determine the value $y$. We shall now see that this distinction is very important, allowing for us to do better than \emph{both} the lower bound for the retrieval problem \cite{alstrup2001optimal} and our lower bound for the tiny-pointer problem (Theorem \ref{thm:lower-bound-variable}). At the same time (almost paradoxically), it is our construction for variable-size tiny pointers that allows for us to get around both of these lower bounds.

\begin{thm}
Consider the relaxed retrieval problem with $k$-bit keys, $v$-bit values, and a maximum capacity of $n$ key/value pairs. Let $r \in [\log^* n]$ be a parameter. There is a solution to the relaxed retrieval problem that uses tiny retrievers of expected size $O(1)$, and that with high probability in $n$: takes constant time per query, takes $O(r)$ time per insertion/deletion, and uses total space $nv + O(n \log^{(r)} n)$ bits. 

Furthermore, if $\log^{(r)} n = \omega(1)$ and $v \le \frac{\log n}{\log^{(r)} n}$, then the space consumption becomes $nv + O(n)$ bits.
\label{thm:retrieve}
\end{thm}

The above theorem comes with an interesting tradeoff curve: constant-time insertions/deletions can achieve a space consumption of, for example, $nv + O(n \log \log \log \log \log n)$ bits, and $O(\log^* n)$-time insertion/deletions can achieve space consumption $nv + O(n)$ bits. Moreover, if $v$ is slightly sub-logarithmic, then even constant-time insertions/deletions can achieve $nv + O(n)$ bits. 

We remark that the tiny retrievers in Theorem \ref{thm:retrieve} are, in fact, variable-size tiny pointers as constructed in Theorem \ref{thm:upper-bound-variable}. They therefore satisfy the doubly-exponential tail inequality given by Theorem \ref{thm:upper-bound-variable}, as well as the concentration inequality given by Proposition \ref{prop:aggregatesize}.

\begin{proof}
We shall make use of Theorem \ref{thm:upper-bound-variable} to construct a \reftable $T$ with $2n$ slots. What makes our application of Theorem \ref{thm:upper-bound-variable} unusual, however, is that we will not store anything in the \refarray (if fact, we need not even allocate space for it). Instead, we will take advantage of the fact that $\Dereference(x, p)$ is a $(1 + \log n)$-bit number that has been uniquely allocated to $x$. 

To implement the operation $\Insert(x, y)$, we call $\Allocate(x)$ to obtain a tiny pointer $p$ of expected size $O(1)$ (note that $p$ will also be our tiny retriever). Define $s_x = \Dereference(x, p)$ to be the slot number in $[2n]$ allocated to $x$. The main property that we will exploit is that $s_x \neq s_{x'}$ for all other $x' \in S$. To complete the $\Insert$ operation, we insert the key/value pair $(s_x, y)$ into a succinct hash table $H$ (whose specifications we will describe later). Queries and deletes are then implemented as follows: $\Query(x, p)$ returns $H[\Dereference(x, p)]$; and $\Delete(x, p)$ deletes key $\Dereference(x, p)$ from $H$ and calls $\Free(x, p)$ on the \reftable $T$.

The correctness of the data structure follows from the fact that, for each $x \in S$ with tiny retriever $p$, $\Dereference(x, p)$ is unique. The \reftable uses space only $O(n)$ bits and supports constant-time operations (with high probability). Thus, to prove the theorem, it remains to analyze the hash table $H$. 

We construct $H$ using the most space-efficient known construction for a hash table \cite{supersuccinct}. If $H$ is storing up to $n$ keys from a universe $U$ and values are $v$ bits, then it supports the following guarantees with high probability: queries are constant-time, insertions/deletions take time $O(r)$, and the total space consumption is 
$$\log \binom{|U|}{n} + nv + O(n\log^{(r)} n)$$
bits. If, in addition, $\log^{(r)} n = \omega(1)$ and $v \le \frac{\log n}{\log^{(r)} n}$, then the space becomes $\log \binom{|U|}{n} + nv + O(n)$ bits.

Our use of tiny pointers ensures that the keys in $H$ are from the very small universe $U = [2n]$. So 
$$\log \binom{|U|}{n} = \log \binom{2n}{n} = O(n)$$
by Stirling's approximation. This completes the proof of the theorem.
\end{proof}

\paragraph{A remark on resizing.} 
In Subsection \ref{trees}, we shall see an application of tiny retrievers to the problem of constructing succinct binary search trees. In this application, we will want to have two relaxed-retrieval data structures whose sizes sum to at most $n$. Here, we can take advantage of the fact that the hash table $H$ used above actually offers a dynamically-resizing guarantee: if, at any given moment, the hash table has size $m$, then it uses space at most
$$\sqrt{n} + \log \binom{2n}{m} + mv + O(m\log^{(r)} n),$$
with high probability in $n$. The full retrieval data structure (consisting of the hash table $H$ and the dereference table $T$) therefore uses space at most 
\begin{align*}
\log \binom{2n}{m} + mv + O(n + m\log^{(r)} n) .
\end{align*}
By Stirling's inequality, this is at most
\begin{align*}
    & m \log n - m \log m + mv + O(n + m\log^{(r)} n).
\end{align*}    
Thus, if we have two relaxed-retrieval data structures, one of size $m_1 \le n$ and one of size $m_2 \le n$, and $m = m_1 + m_2 = \Theta(n)$, then their total space consumption will be at most
\begin{align*}
& (m_1 + m_2) \log n - m_1 \log m_1 - m_2 \log m_2 + (m_1 + m_2) v + O((m_1 + m_2)\log^{(r)} n) \\
=\, &  m \log n - m_1 \log m_1 - m_2 \log m_2 + m v + O(m\log^{(r)} n).
\end{align*}
By Jensen's inequality,  $m_1 \log m_1 + m_2 \log m_2 \ge (m_1 + m_2) \log \frac{m_1 + m_2}{2} = m \log \frac{m}{2} = m \log n - O(n)$. Thus the total space is at most
\begin{align*}
& m \log n - (m \log n - O(n)) + m v + O(m\log^{(r)} n) \\
& = m v + O(m\log^{(r)} n) \\
& = m v + O(m\log^{(r)} m) \\
\end{align*}
This, of course, is the same bound that we get for a single relaxed-retriever data structure of size $m$.

The reason that this matters is that it allows for a simple way to perform dynamic resizing: every time that the size $m$ of a data structure changes by a factor of two, we move all of the elements in the current relaxed-retrieval data structure $D_1$ into a new relaxed-retrieval data structure $D_2$ (parameterized as having capacity $n = \Theta(m)$ based on the new value of $m$). As we move elements from $D_1$ to $D_2$, the total space consumption of $D_1$ and $D_2$ will continue to be $ m v + O(m\log^{(r)} m)$ bits. Note that, to move an element from $D_1$ to $D_2$, we will need to generate a new tiny retriever for that element (since we are deleting the element from $D_1$ and inserting it into $D_2$). In our binary-search-tree application, this will be easy to do by simply running through all of the elements and relocating them one by one. Furthermore, since the work of constructing $D_2$ can be spread across $\Theta(n)$ operations, it can be achieved at a cost of $O(r)$ per insertion/deletion.

\subsection{Succinct Binary Search Trees}\label{trees}

Our second application is a black-box approach for transforming dynamic binary search trees into succinct data structures. If there are $ n $ elements in the succinct search tree, each of which is $k$ bits long, then the size of the succinct search tree will be at most $nk + O(n + n\log^{(r)} n )$ bits, where $r > 0$ is an arbitrary parameter. Path traversals in the tree incur only a constant-factor overhead, and modifications to the tree incur only an $O(r)$-factor overhead. 

An advantage of our approach is that it can be applied to rotation-based search trees. This includes, for example, red-black trees \cite{GuibasSe78}, splay trees \cite{SleatorTa85-splay}, etc. If the dynamic-optimality conjecture \cite{SleatorTa85-splay} is true, meaning that the splay tree is dynamically optimal, then our succinct splay tree is also dynamically optimal when $r = O(1)$.

\begin{thm}
Consider any binary search tree storing $a$-bit keys and $b$-bit values, where every node is associated with a distinct key, and where each node has pointers to its children. For any $r > 0$, the tree can be implemented to offer the following guarantees with high probability in the tree size $n$: the tree takes space $na + nb + O(n + n\log^{(r)} n)$ bits, traversals from parents to children take time $O(1)$, and modifications to the tree (i.e., adding or removing an edge) take time $O(r)$. 
\label{thm:trees}
\end{thm}

We remark that, information theoretically, the tree use consume $n(a + b)$ bits of space. And since the keys are distinct, $na = \Omega(n \log n)$. Thus, for any $r > 1$, the search tree above is succinct.

\begin{proof}
We will make use of our solution to the relaxed retrieval problem (Theorem \ref{thm:retrieve}). However, the key/value pairs $(x, y)$ that we will store in the relaxed-retrieval data structure will be a bit unusual in that $y$ will take the following form: $y$ contains $x$'s $b$-bit value, along with two tiny retrievers $r_1$ and $r_2$. Since $r_1$ and $r_2$ are themselves variable-length tiny pointers of expected size $O(1)$, this means that $y$ is also variable-length. On the other hand, the relaxed-retrieval data structure is designed for \emph{fixed-length} values. Fortunately, we can store the tiny retrievers $r_1$ and $r_2$ with the following method. Recall that, in our construction for the relaxed retrieval problem, we create a \reftable with $2n$ slots, but we do not actually store anything in the \reftable's \refarray. We now change this so that the \refarray is a value array with $2n$ slots that stores the tiny retrievers $r_1$ and $r_2$ for each item in the \reftable (so, if $p$ is the tiny pointer for $x$, then $r_1, r_2$ are in the $\Dereference(x, p)$-th position of the value array). Using the chunked pointer storage technique, we can ensure that the total size of the value array is $O(n)$ bits, even though the pointers that it stores are variable length.

We now describe our encoding of the binary search tree: Each node in the search tree stores the key-value pair $(x, y)$ corresponding to that node along with two tiny retrievers $r_1$ and $r_2$. The tiny retriever $r_1$ is for the left child and uses $x \circ 0$ as its key (so $\Query(x \circ 0, r_1)$ returns the left child of $x$), and the tiny retriever $r_2$ is for the right child and uses $x \circ 1$ as its key (so $\Query(x \circ 1, r_1)$ returns the right child of $x$). Note that, if the left child (resp. right child) does not exist, then we simply set $r_1$ (resp. $r_2$) to null. 

Let us begin by assuming that our binary search tree has a fixed capacity of $n$ keys/values, so we can use a relaxed-retrieval data structure with capacity $n$. Then our relaxed-retrieval data structure uses $na + nb + O(n + n\log^{(r)} n)$ bits. Navigating from a node to its child takes time $O(1)$ (since it requires a single query to the relaxed-retrieval data structure) and adding/removing an edge $(x, z)$ from a node $x$ to a child $z$ takes time $O(r)$, with high probability, since it requires only a single insert/delete to the relaxed-retrieval data structure; importantly, if $z$ is the root of some subtree, the act of setting $z$ to be $x$'s child \emph{does not} require any nodes besides $z$ to inserted/deleted in the relaxed-retrieval data structure. 

Finally, let us modify our data structure so that it dynamically resizes as a function of the current number $n$ of key/value pairs. For this, we can simply use the resizing approach outlined in Section \ref{retrieval}. Every time that $n$ changes by a constant factor, we rebuild the relaxed-retrieval data structure to have capacity $\Theta(n)$ for the new value of $n$. (Note that this does not require us to rebuild the tree; it just requires us to update the tiny retrievers used in each node.) For each relaxed retriever in the binary search tree, we can store an extra bit indicating which of the two relaxed-retrieval data structures it uses---this preserves correctness. As observed in Section \ref{retrieval} the act of moving items from the old relaxed-retrieval data structure to the new one does not violate our desired space guarantee: the total number of bits used by our search tree remains $na + nb + O(n + n\log^{(r)} n)$ at all times.  And, by spreading the work of rebuilding the relaxed-retrieval data structure across $\Theta(n)$ operations, we maintain the property that each edge insertion/deletion takes time $O(r)$. Thus the theorem is proven.
\end{proof}

\subsection{Space-Efficient Stable Dictionaries}\label{stable}
Using tiny pointers, we give a black-box approach for transforming any fixed-capacity key-value dictionary into a  \defn{stable} dictionary, meaning that the position in which a value is stored never changes after the value is inserted. If the original dictionary stored $ k $-bit values, then the new dictionary also stores $ k $-bit values, and uses at most $O\left(\log k\right) $ extra bits of space per value than the original data structure.

\begin{thm}
Consider a fixed-capacity key-value dictionary data structure $T$ that stores its values in a value array of some size $m$. Let $ v $ denote the size of each value in bits.

It is possible to construct a new data structure $T'$ with the same operations and asymptotics (with high probability) as $T$, but with the additional property that $ T' $ is stable. Moreover, the total space consumed by $ T' $ is guaranteed (with high probability in $ m $) to be at most $ O (m \log v)$ more bits than $T$.
\label{thm:stable}
\end{thm}
\begin{proof}
To construct $ T' $, we simply replace the value array for $ T $ with an array of $ m $ tiny pointers, each of size $\Theta (\log v) $ bits. (If  $\log v <\log\log\log n $, then the chunked-storage technique can be used to handle the fact that different tiny pointers have different sizes.) The tiny pointers point into a dereference table of size $(1 + 1/v)m$ that stores the $m$ $ v $-bit values. (So the load factor is $1 - \Theta(1/v)$.) If a tiny pointer points at the value $y$ corresponding to a key $x$, then the tiny pointer uses $ x $ as its key. This ensures stability, since even if the location in which the tiny pointer is stored changes, the tiny pointer does not have to change  (and the value $ y $ does not have to move). 

The array of tiny pointers consumes $ O (m\log v) $ space. Whereas the value array in $ T $ consumes $ mv $ bits, the dereference table in $ T' $ consumes $ (1+1/v) mv $ bits, which is only $ O (m) $ more bits then used in $ T $. Thus the claim on space efficiency is proven. Since tiny pointers only add constant time per access/modification of the value, the asymptotics are (with high probability in $ m $) the same for both $ T $ and $ T' $.
\end{proof}

\subsection{Space-Efficient Dictionaries with Variable-Size Values}
\label{var}
Our fourth application is a black-box approach for transforming any key-value dictionary (designed to store fixed-size values) into a dictionary that can store different-sized values for different keys. The resulting data structure offers the following remarkable guarantee on space efficiency. Let $\log^{(r)} n = \log \log \cdots \log n$ denote the $ r $-th iterated logarithm of $ n $. Let $r$ be a positive constant of our choice, and let $m$ be the number of entries in the value array used by the original dictionary (at some given moment). The new dictionary, which allows for values to be arbitrary lengths, replaces the value array for $ T $ with a data structure that consumes at most
$$O(m \log^{(r)} m) + \sum_{i = 1 }^{m} (v_i + O(\log v_i))$$
bits, where $ v_1, v_2, \ldots, v_m $ denote the lengths in bits of the values being stored. 

\begin{thm}
Consider a key-value dictionary data structure $T$ that stores its values in a value array, and that is designed to store fixed-length keys. Let $ r $ be a positive constant of our choice. 

It is possible to construct a new data structure $T'$ with the same operations and asymptotics (with high probability) as $T$, but with the additional property that $ T' $ can store values of arbitrary lengths (up to $O(1)$ machine words).

At any given moment, if $ T $ would have been using a value array of size $ m $, and the machine word size $w$ satisfies $w \le m^{o(1)}$, then the total space consumed by $ T' $ to implement the value array is guaranteed (with high probability in $ m $) to be at most 
\begin{equation}
O(m \log^{(r)} m) + \sum_{i = 1 }^{m} (v_i + O(\log v_i))
\label{variablelength}
\end{equation}
bits, where $ v_1, v_2, \ldots$ are the sizes of the values.
\label{thm:var}
\end{thm}

We remark that the limitation on value-size to be $ O (1) $ machine words is simply so that each value can be written/read in constant time, that way it is easy to discuss how the asymptotics of $T$ and $T'$ compare. The same techniques work for even larger values without modification, as long as one is willing to spend the necessary time to read/write values that are of super-constant size.

\begin{proof}[Proof of Theorem \ref{variablelength}]
Values in $ T' $ are stored with up to $ r $ levels of indirection. If a value is $ k $ bits, then it is pointed at by a tiny pointer $p_1$ of size $ O (\log k) $ bits. The tiny pointer $ p_1 $ is, in turn, pointed at by a tiny pointer $ p_2 $ of size $ O (\log\log k) $ bits, and so on, with pointers of size $ O (\log\log\log k), O (\log\log\log\log k), \ldots, O (\log ^ {(r) } n) $. That is, every value is stored at the end of a linked list of length $ O (1) $, where the base pointer of the linked list is $O (\log ^ {(r) } n)$ bits, and each subsequent pointer is exponentially larger than the previous one.

For each tiny pointer of some size $ j $ in the data structure, we must also store $O(j)$ extra bits of information indicating (a) whether the tiny pointer is pointing at another tiny pointer or at a final value, and (b) what the size is of the tiny-pointer/value being pointed at. Throughout the rest of the proof, we will count these $ O (j) $ extra bits as being part of the size of the tiny pointer.

Since there are both values and tiny pointers of many different sizes, we must use a different dereference table for each size-class of tiny-pointer and the different dereference table for each size-class of values being stored. (Note that the dereference tables storing tiny pointers may need to use the chunked-storage technique to handle variable-sized tiny pointers, so the same dereference table should not be used to store both tiny pointers and values.)

The problem of dynamically resizing all of the dereference tables simultaneously is slightly tricky. Consider a dereference table $A$ (to $A$ could also be the value array) that stores $ j $-bit tiny pointers for some $ j $. There are $ K = 2 ^ {\Theta(j) } $ different dereference tables $B_1, B_2, \ldots, B_k$ that these tiny pointers can point into (depending on the size of the object being pointed at, and whether the object is a tiny pointer or a value). Each $B_i$ must individually be dynamically resized. We will maintain what we call the  \defn{dynamic-sizing invariant}, which guarantees that each $B_i$ is either (a) at a load factor $1 - O(1/j')$, where $ j' $ is the size of the objects stored in $B_i$, or (b) is at most a $o(1 / (Kj))$-fraction the size (in bits) of $A$. 

To implement the dynamic-sizing invariant, we dynamically resize each $ B_i $ using zone-aggregated resizing (recall from Section \ref{applicationprelim} that this means $ B_i $ is broken into multiple components, and each component is occasionally rebuilt so that its size either doubles or halves). To allow for components of each $ B_i $ to be rebuilt efficiently, we break $ A $ into zones of size $\poly(K)$, meaning by \eqref{eq:rebuild} from Section \ref{applicationprelim} that a given component (of some $B_i$) consisting of $s$ entries can be rebuilt in time
$$|A| / \poly(K) + s,$$
where $|A|$ is the number of entries in $A$. We perform dynamic resizing on $ B_i $ differently depending on whether it is very small (its components contain fewer than $|A| / \poly(K) $ elements each) or not: 
\begin{itemize}
\item If the components contain $s = \Omega(|A| / \poly(K))$ elements each, then we perform zone-aggregated resizing (exactly as in Section \ref{applicationprelim}) to keep $ B_i $ at a load factor $1 - O(1/j')$, where $ j' $ is the size of the objects stored in $B_i$. In this case, the time needed to rebuild a component of size $ s $ is $\Theta (s ) $, so the dynamic resizing of $ B_i $ can be deamortized to take $ O (1) $ time per operation (on $B_i$). Note that, here, $ B_i $ is in case (a) of the dynamic-resizing invariant.
\item If the components contain fewer than $|A| / \poly(K)$ elements each, then we perform zone-aggregated resizing to keep each component of $ B_i$ at a capacity of $\Theta(|A| / \poly(K))$ (even as $|A|$ changes over time, and \emph{regardless} of whether the number of elements per component may be significantly smaller than $|A| / \poly(K)$).  Note that, here, $ B_i $ is in case (b) of the dynamic-resizing invariant.

When $ B_i $ is in this regime, we cannot amortize the work spent rebuilding $ B_i $ to the operations that are performed on $B_i$. Instead, we spread out the work spent rebuilding components of $ B_i $ in the following way: for every $\Theta (K) $ work that is spent on $ A $ we also spend $O(1)$ time on resizing $B_i$. Since $ B_i $ is more than a factor of $ K $ smaller than $ A $, this is sufficient time to keep $ B_i $ in a state where each component has capacity $\Theta(|A| / \poly(K))$. 

From the perspective of $ A $, every time that we spend constant time on insertions/deletions/rebuilding $ A $, we also may spend constant time performing rebuild-work on one of the $B_i$s (which, in turn, may recursively lead us to spend constant time on rebuilding one of the dereference tables pointed at by $ B_i $, etc.). Importantly, since chains of tiny pointers are at most $r \le O(1)$ long, the time spent on rebuilds only introduces a constant-factor overhead on running time per operation.
\end{itemize}

The resizing approach described above guarantees the dynamic-sizing invariant while incurring only a constant-factor time overhead per operation. Next we use the invariant to bound the space consumption of $ T' $.
The dereference tables $B_i$ in case (a) are implemented space-efficiently enough that the empty slots in them take negligible space compared to the actual objects stored in them (i.e., the empty slots add $ O (1) $ bits per object), and the dereference tables $B_i$ in case (b) are small enough that they take negligible space compared to the size of the parent dereference table $A$ (i.e., they cumulatively add $ o(1) $ bits per slot in $A$). It follows that the total space consumed by dereference tables will be at most the sum of the sizes of the objects being stored in the dereference tables, plus $O(1)$ bits per object; this, in turn, means that the space used by $ T' $ to store values/tiny pointers is given by \eqref{variablelength}.

Next, we bound the time-overhead of $ T' $ when compared to $ T $. We have already shown that the time-overhead of performing dynamic-resizing on dereference tables is $ O (1) $ per operation. Since values are stored with at most $ r = O (1) $ levels of indirection, the time needed to access/modify a value is also $ O (1) $. Thus $ T' $ has the same time asymptotics as $ T $.

Finally, we argue that the dereference tables used by $ T' $ succeed at their allocations with high probability.\footnote{There are many different ways that one could handle allocation failures, including, for example, performing batch-rebuilds of the data structure.} There are several approaches that we could take to doing this; the simplest is to just add one more modification to how we perform dereference-table resizing: whenever a dereference table gets down to size $\Theta(\sqrt{m})$, we do not ever resize it to be any smaller.\footnote{However, since $m$ may dynamically change over time, we do need to spend constant time per operation resizing dereference tables of size $\Theta(\sqrt{m})$ so that they stay size $\Theta(\sqrt{m})$ as $m$ changes.} This means that some dereference tables could be very sparse, containing $\sqrt{m}$ slots, but containing far fewer elements. Since there are only $O(w) = m^{o(1)}$ different dereference tables (recall that $w$ is the machine-word size), the net space consumption of the dereference tables of size $\Theta(\sqrt{m})$ is $o(m)$ bits. The fact that every dereference table has size at least $\Omega(\sqrt{m})$ means that all of the dereference tables offer high probability guarantees, as desired.
\end{proof}

\subsection{An Optimal Internal-Memory Stash}\label{stash}
Our final application of tiny pointers revisits one of the oldest problems in external-memory data structures: the problem of maintaining a small internal-memory \defn{stash} that allows for one to locate where elements reside in a large external-memory data structure. 

The problem can be formalized as follows. We must store a dynamically changing set $S$ of up to $n$ key-value pairs, where each key-value pair can be stored in one machine word, and where each key is unique. We are given an \defn{external memory} consisting of $ (1+\epsilon) n $ machine words, where the key-value pairs $S$ are to be stored. In addition to storing key-value pairs in external memory, we must maintain a small internal-memory data structure $X$, which we will refer to as the \defn{stash}, that supports the following operations: 
\begin{itemize}
\item \textbf{Query$ (k) $:} Using only information in the stash data structure, returns the position in external memory where the key $ k $ and its corresponding value $ v $ are stored.
\item \textbf{Insert$ (k, v) $:} Inserts the key-value pair $ (k, v) $, placing the pair somewhere in external memory, and updating the stash. 
\item \textbf{Delete$ (k, v) $:} Removes the key/value pair $(k, v)$ from the external-memory array, and updates the stash.
\end{itemize}

The important feature of a stash is that queries can be completed with a single access to external memory. On the other hand, in order for a stash to be useful, several other objectives must be achieved:
\begin{itemize}
\item \textbf{Compactness:} The stash $X$ needs to be as small as possible, that way it can fit into an internal memory with limited size.
\item \textbf{Efficient inserts and deletes:} Although a stash prioritizes queries, insertions and deletions should ideally also require only $O(1)$ accesses/modifications to external memory.
\item \textbf{RAM efficiency:} Finally, so that computational overhead does not become a bottleneck, the operations on a stash should be as efficient as possible in the RAM model, ideally taking time $O(1)$.
\end{itemize}

A concrete example of a stash that is used in real-world systems is the \defn{page table} \cite{IntelManual,AMDManual,BenderBhCo21}, which is an operating-system-level dictionary that maps virtual page addresses to where their corresponding physical pages reside in memory. The page table is accessed for every address translation, so it is performance critical and thus highly optimized. Additionally, it is important that the page table be space-efficient, so that it may be effectively cached in the processor cache hierarchy. 
Note that, although page tables get to select where physical pages reside in memory, they do not get to move physical pages that have already been placed; thus any stash that is used as a page table must also be stable. For this reason, past work \cite{Larson1, Larson2, Larson3} has typically included stability as an additional criterion for a stash.

Work on designing space-efficient and time-efficient stashes dates back to the late 1980s \cite{Larson2, Larson3, Larson1}. The best-known theoretical results are due to Gonnet and Larson \cite{Larson1}, who give a stable stash that uses only $O(n \log \epsilon^{-1})$ bits. A remarkable consequence of this is that, when $\epsilon = \Theta(1)$, it is possible to construct a stash using only $O(n)$ bits. 

Gonnet and Larson's result comes with several significant drawbacks, however \cite{Larson1}, which have proven difficult to fix. First, due to its reliance on stable uniform probing \cite{Larson4} as a mechanism for determining where keys/values should reside, the stash only offers provable guarantees in the setting where insertions/deletions are performed \emph{randomly}. Second, the data structure is not constant-time in the RAM model, instead taking expected time $\Theta(\epsilon^{-1})$.

Using tiny pointers, we show that modern techniques for constructing filters can easily be adapted in order to construct a stable stash of size $O(n \log \epsilon^{-1})$ bits that supports constant-time operations in the RAM model (with high probability) and that supports arbitrary sequences of insertions/deletions/queries.

\begin{thm}
It is possible to construct a stable stash that supports constant-time operations in the RAM model, that stores up to $m$ keys/values in an external-memory array of size $ (1+\epsilon) m$, and that uses only $ O (m\log\epsilon ^ { -1 }) $ bits of internal-memory space. All of the guarantees for the stash hold with high probability in $ m $.
\label{thm:stash}
\end{thm}
\begin{proof}
The starting point for our design is the adaptive filter of Bender et al. \cite{BenderFaGo18}. 
Like a stash, their filter is a space-efficient internal-memory data structure that summarizes the state of an external-memory key-value dictionary. Unlike a stash, their filter does not indicate where in external memory each key/value is stored. Instead, the filter answers containment queries with the following guarantee: each positive query is guaranteed to return true, and each negative query is guaranteed to return false with probability at least $ 1 -\epsilon $ (for some parameter $\epsilon $). The size of their internal-memory data structure is only $ (1+ o (1)) m\log\epsilon ^ { -1 } = O(m \log \epsilon^{-1})$ bits, where $ m $ is the capacity of the filter.\footnote{In fact, their data structures also dynamically resizable, but for our application that will not be necessary.} 

The basic idea behind the adaptive filter of \cite{BenderFaGo18} is to store a \defn{fingerprint} for each key $x$, where each fingerprint is taken to be some prefix of the hash $h(x)$. Different keys have different-length fingerprints, and the invariant maintained by the filter is that no fingerprint is a prefix of any other fingerprint. To maintain this invariant while also keeping the fingerprints as small as possible, the filter will sometimes change the lengths of $O(1)$ different fingerprints during a given insertion/deletion; to change the length of a fingerprint, the key corresponding to that fingerprint must first be fetched from external memory, that way the hash $h(x)$ of that key can be recomputed.\footnote{The original data structure also sometimes updates the lengths of fingerprints during negative queries, but such updates are not needed for the purposes of our data structure.} 

The fingerprints in the filter are stored as follows. The first $\lg n$ bits of each fingerprint are called the \defn{quotient}, and these bits are used to assign the key to one of $n$ bins; importantly, the fact that the bin-choice encodes the quotient of each of the keys in the bin means that the data structure does not have to explicitly store the quotients of the fingerprints. The next $\log\epsilon ^ { -1 } $ bits of each fingerprint are called the  \defn{baseline bits}, and these bits are included for every fingerprint in the data structure. Finally, any subsequent bits in a fingerprint are called the \defn{adaptivity bits}, and these bits are added/removed in order to maintain the prefix-freeness invariant. A central piece of \cite{BenderFaGo18}'s analysis is to show that there are only $O(m)$ adaptivity bits in total, and that these bits can be stored efficiently.

We now describe how to modify the filter to be a stash. In addition to storing a fingerprint for each key, we now also store a tiny pointer with expected size $\Theta(\log \epsilon ^ { -1 }) $. These tiny pointers are easy to store, since the filter has already made room for $\log\epsilon ^ { -1 } $ baseline bits for each key. Of course, different tiny pointers may have different lengths, but this issue can easily be resolved by either using the chunked-storage technique described in Section \ref{applicationprelim} (or by adapting the techniques already used in \cite{BenderFaGo18} to handle variable-length fingerprints).

One minor difficulty is that the filter assumes access to an external-memory dictionary (rather than just a dereference table) that way it can lookup keys in order to modify their fingerprints. In the case of our stash, however, these lookups can easily be performed using the tiny pointers that are already stored, so one does not need a full dictionary in external memory.

The fact that the tiny pointers have size $\Theta(\log \epsilon ^ { -1 }) $ means that external memory can be implemented as a dereference table with load factor $ 1-\epsilon $. The fact that the original adaptive filter supported constant-time operations (with high probability in $m$) translates to the stash also supporting constant-time operations. And the fact that the original adaptive filter used space $O(m \log \epsilon^{-1})$ bits in internal memory also translates the same guarantee for the stash. Thus the theorem is proven.

\end{proof}

%% file: balls-n-bins.tex
\section{Dynamic Balls and Bins}
\label{sec:balls-n-bins}

In this section, we reinterpret our tiny-pointer constructions as balls-and-bins schemes in order to improve the state of the art
for the classic dynamic load balancing problem.

In the dynamic load-balancing problem, there is a system of $n$ bins and a large universe $U$ of balls. Balls are inserted and deleted (and sometimes reinserted)
over time by an oblivious adversary, so that the total number of balls in the system never exceeds $m = hn$ for some parameter $h$. Whenever a ball $x$ is inserted, it must be placed in one of $d$ bins from among $\h_1(x), \ldots, \h_d(x)$, where $\h_i()$ is some hash function from balls to bins. Once a ball is placed in a bin, it cannot be moved until it is deleted. The goal of the \defn{dynamic load-balancing problem} is to assign balls to bins in order to achieve the smallest maximum-load possible (i.e., to minimize the number of balls in the fullest bins). We refer to the special case where balls can be inserted and deleted but not reinserted as the \defn{semi-dynamic load-balancing problem.}

There are two classic solutions to the problem. The first is \single balls-to-bins assignment: we set $d = 1$ and just place each $x$ in $h_1(x)$. 
The second is $\leftd{d}$ balls-to-bins assignment: divide the bins into $d$ groups so that each $h_i$ is uniform into the $i$-th group; when inserting $x$, pick the bin $h_i(x)$ 
with the smallest load, and break ties by minimizing $i$.

\single's behavior history independent, in that the maximum load at any time only depends on which balls are present, and not the history of their arrival.  The maximum load is then completely characterized by standard Chernoff bounds.

$\leftd{d}$, on the other hand, is highly history dependent. The first time that a ball $x$ is inserted, the hashes $\h_1(x), \ldots, \h_d(x)$ are independent of the system state, but if a ball $x$ is ever deleted and then later \emph{reinserted}, then
the past insertion of $x$ can have long-term side effects on the system state meaning that the state is not necessarily independent of $\h_1(x), \ldots, \h_d(x)$.

In the insertion-only setting (i.e., balls are not deleted), $\leftd{d}$ offers a celebrated  bound~\cite{Vocking03} of 
\begin{equation}
h+ \frac{\log\log n}{d\log\phi_d} + O(1)
\label{eq:staticV}
\end{equation}
on maximum load, where $\phi_d$ is the generalized golden ratio. In the dynamic setting, $\leftd{d}$ has proven to be 
significantly more difficult to analyze. The original analysis of $\leftd{d}$ by V\"ocking \cite{Vocking03} can be used to achieve a bound of
\begin{equation}
O(h) + \frac{\log\log n}{d\log\phi_d}
\label{eq:semiV}
\end{equation}
for the semi-dynamic setting, but as Woefel observed\cite{DBLP:conf/soda/Woelfel06}, the same argument does not apply directly to the fully dynamic setting.\footnote{The difficulty has to do with the analysis of the leaves in the witness tree, and is easy to describe in the case where $h = 1$. To analyze a leaf ball $x$, the original analysis uses Markov's inequality to deduce that each of $x$'s $d$ bins has at most a $1/3$ probability of having $3$ or more balls, and the analysis concludes that the probability of all $d$ bins containing $3$ or more balls is at most $1/3^d$. This same analysis does not apply in the fully dynamic setting since it would need the state of the system of to be independent of $x$'s hash functions $\h_1(x), \ldots, \h_d(x)$, which is not the case due to subtle history dependencies in the system's state.}  He shows how modify V\"ocking's proof to achieve a bound of 
\begin{equation}
O(d) + \frac{\log\log n}{d\log\phi_d}
\label{eq:weofel}
\end{equation}
in the setting where $h = 1$. In general, when $h > 0$, Woefel's argument yields a bound of 
\begin{equation}
O(1 + hd) + \frac{\log\log n}{d\log\phi_d},
\label{eq:semiV2}
\end{equation}
which has remained the state of the art.

The bound \eqref{eq:semiV2} is most interesting in the case where $h$ is relatively small, that is, $h = o(\log n)$. Here, \eqref{eq:semiV2} can be significantly better than the $\Theta(\log n / \log \log n)$ bound that would be achieved by \single. Of course, the question remains as to whether there exists a balls-to-bins scheme that achieves a better bound. We answer this question in the affirmative, by giving a bin-selection rule with $d + 1$ hash functions that achieves maximum load
\begin{equation}
h + \frac{\log\log n}{d\log\phi_d} + O(\sqrt{h \log (hd)}).
\label{eq:semiV3}
\end{equation}
We remark that, even when $h$ is a constant, this bound improves the dependence on $d$ from $O(d)$ to $O(\sqrt{\log d})$.

Our rule, which we call \defn{$\iceberg{d}$} is a hybrid of \single and $\leftd{d}$. This rule is closely related to the rule that we used in Section \ref{sec:fixed} for constructing fixed-size tiny pointers.

The rest of the section proceeds as follows. We begin in Subsection~\ref{ssec:iceberglemma} by proving a useful technical lemma. In Subsection \ref{ssec:icebergd}, we present and analyze $\iceberg{d}$. Finally, in Subsection \ref{ssec:probe}, we reinterpret our variable-size tiny-pointer construction as a result about probe-complexity of balls-and-bins schemes with bins of capacity $1$; in particular, we give the first dynamic ball-allocation scheme to offer $\poly(\delta^{-1})$ average probe complexity in the setting where there are up to $(1 - \delta)n$ balls present in the system at a time.

\input{iceberg-lemma}

\subsection{$\iceberg{d}$}\label{ssec:icebergd}

We now present the $\iceberg{d}$ balls-in-bins rule. Let $ n $ be the number of bins, let $ h n$ be the maximum number of balls allowed to be present at any given moment, and let $ d > 1$ be a parameter. Partitioning the bins into $ d $ equal-size sets $ S_1,\ldots, S_d $. Let $ g $ be a hash function mapping balls uniformly at random to bins, and let $ h_1,\ldots, h_d $ be hash functions such that each $ h_i $ maps balls uniformly at random to a random bin in $ S_i $.

We shall have three types of balls: level-one balls, level-two balls, and level-three balls. Each level-one ball $ x $ will reside in bin $ g (x) $, each level-two ball $ x $ will reside in one of bins $ h_1 (x),\ldots, h_d (x) $, and each level-three ball $ x $ will reside in bin 1 (but, at any given moment, the number of level-three balls will be zero w.h.p.). 

Set $\tau = c\sqrt{h \log (hd)}$ for some sufficiently large positive constant $ c $. We shall also keep track of a variable $ q $ counting the number of level-two balls present at any given moment. 

The procedure for inserting a ball $ x $ is as follows. If bin $ g (x) $ contains $ h +\tau $ level-one balls or fewer, then we place $ x $ in bin $ g (x) $, and we classify $ x $ as a level-one ball. Otherwise, we check whether $q < n / d$. If $q < n / d$, then we examine bins $ h_1 (x),\ldots, h_d(x)$, and we place $ x $ as a level-two ball into whichever bin $ h_i (x) $ contains the fewest level-two balls  (breaking ties towards the smallest $ i $). Finally, if $q \ge n / d$, then we place $ x $ as a level-three ball into bin 1.

\begin{thm}
Suppose $1 \le  h \le n ^ { o (1) } $ and $1 <  d \le n ^ { o (1) } $. Suppose balls are inserted/deleted/reinserted into $n$ bins over time (by an oblivious adversary) according to $\iceberg{d}$ rule, with no more than $ hn $ balls present at a time. Then, w.h.p.\ in $n$, at any given moment, the number of balls in the fullest bin is $h + \frac{\log\log n}{d\log\phi_d} + O(\sqrt{h \log (hd)})$.
\label{thm:balls}
\end{thm}
\begin{proof}
Each bin deterministically contains at most $h + \tau = h + O(\sqrt{h \log (hd)})$ level-one balls. Thus, it suffices to bound the number of level-two and level-three balls in each bin by $ \frac{\log\log n}{d\log\phi_d} + O(1)$. 

The number $ q $ of level-two balls in the entire system is deterministically at most $n / d$ at any given moment. In other words, the level-two balls are placed according to the $\leftd{d}$ rule with $h'n$ balls, where $h' = 1/d$. Thus we can apply \eqref{eq:semiV2} to deduce that, w.h.p., the maximum number of such balls per bin is
$$O(1 + h'd) + \frac{\log\log n}{d\log\phi_d} = \frac{\log\log n}{d\log\phi_d} + O(1),$$
Note that, in this application of \eqref{eq:semiV2}, we are using Woefel's analysis \cite{DBLP:conf/soda/Woelfel06} of $\leftd{d}$ in a somewhat unusual parameter regime; that is, the analysis is intended primarily to be used in the regime $h' \ge 1$ (and Woefel's result was only explicitly stated for $h' = \Theta(1)$), but we are taking advantage of the fact that the analysis also holds for $h' = o(1)$ without modification. 

We complete the proof by showing that, w.h.p., The number of level-three bins is zero. By Lemma \ref{lem:iceberg}, the number $q$ of level-two balls satisfies $q < n / h$ (at any given moment) with probability at least $1 - \exp(-n / \poly(hd))$, which by the assumption $h, d, \le n^{o(1)}$ is at least $1 - 1 /\poly(n)$. It follows that each individual ball insertion has probability at most $1/ \poly(n)$ of being level-three. Taking a union bound over all of the balls in the system, the probability that any of them are level-three is $1 / \poly(n)$, as desired.
\end{proof}

\subsection{Assigning Balls to Capacity-1 Bins with Low Average Probe Complexity} \label{ssec:probe}

Our final result of the section considers a dynamic balls-and-bins game in which there are $ n $ bins each with capacity $ 1 $, and at most $ (1 -\delta) n $ balls are present at a time. Each ball $ x $ has a predetermined (infinite) sequence $ h_1 (x), h_2 (x), \ldots$ of bins where it can reside, and we wish to minimize the \defn{probe complexity} of each ball $x$, which is defined to be the smallest $i$ such that ball $x$ is in bin $h_i(x)$. Since we are in the dynamic setting, the same ball may be inserted, deleted, and reinserted many times.

First note that, in the insertion-only setting, it is easy to achieve probe average complexity $O(\delta^{-1})$ by simply using uniform probing, which sets each $ h_i (x) $ to be random, and places each ball $ x $ into the first available slot in the seqeunce $ h_1 (x), h_2 (x), \ldots$. In the dynamic setting, however, there is not yet any known bin-assignment scheme that achieves average probe complexity $\poly(\delta^{-1})$ (for example, uniform probing has only successfully been analyzed in the random-deletions setting \cite{Larson4}, and the analysis of linear probing without moving elements around remains an open problem \cite{sandersstability}). 

We now construct a bin-assignment scheme that achieves average probe complexity $\poly(\delta^{-1})$.
\begin{thm}
Suppose $\delta = 1 / n^{o(1)}$. There exists a bin-assignment scheme that supports arbitrary ball insertions/deletions/reinsertions, and guarantees an expected probe complexity of $O(\poly(\delta^{-1})$ for each ball in the system.
\label{thm:probecomplexity}
\end{thm}
\begin{proof}
Consider a variable-size-tiny-pointer dereference table with $ n $ slots and load factor $1 -\delta$. For each ball $ x $ and each $i \in \mathbb{N}$, define $h_i(x) = \textsc{Dereference}(x, i)$. To assign a ball $ x $ to a bin, we call the function $i = \textsc{Allocate}(x)$, and place $x$ into bin $h_i(x) =  \textsc{Dereference}(x, i)$. To delete a ball $ x $, we call $\textsc{Free}(x, i)$ in order to deallocate the appropriate slot in the dereference table.

Let $c > 0$ be a sufficiently small positive constant. 
By Theorem \ref{thm:upper-bound-variable}, each ball $x$ gets assigned to a bin $h_i(x)$ where $i$ (which is the tiny pointer returned by $\textsc{Allocate}(x)$) is 
$$O(\log \delta^{-1} + P)$$ 
bits for some random variable $P$ satisfying $\Pr[P \ge j] \le O\left(2^{-2^{cj}}\right)$.
It follows that $\Pr[i \ge \poly(\delta^{-1})k] \le O(2^{-2^{c \log k}}) = O(2^{k^c})$, and hence that the expected probe complexity of each ball $x$ is $\poly(\delta^{-1})$. 
\end{proof}

%% file: iceberg-lemma.tex

\subsection{A Useful Lemma}\label{ssec:iceberglemma}

This section proves a generalization of a technical lemma introduced by a subset of the current authors in recent work on space-efficient hash tables~\cite{iceberg}. The new lemma extends the original to a wider parameter regime. We also take a different combinatorial approach than in the previous paper, resulting in a simpler proof that reveals an interesting relationship between the lemma and Talagrand's inequality.

Consider a dynamic balls-and-bins game with $n$ bins and at most $m = hn$ balls at all times, that are placed with the $\single$ rule.  Whenever a ball is thrown into a bin, if the bin contains $h + \tau$ or more balls, then the ball is labeled as \defn{$\tau$-exposed} (and the label persists until the ball is next deleted). 

\begin{lem}\label{lem:iceberg}
Suppose $1 \le \tau \le h $. At any fixed point in time, the number of $\tau$-exposed balls is $\poly(h) \cdot ne^{- \tau^2/ (3h)}$ with probability $1 - \exp(-\Omega(me^{-\tau^2/(3h)}))$.
\end{lem}

Our proof of the lemma will make use of  a variant of Talagrand's inequality~\cite[Chapter 12]{Talagrand}:
\begin{thm}[Talagrand's inequality]\label{thm:talagrand}
Let $X_1, \dots, X_n$ be $n$ independent random variables from an arbitrary domain. Let $F$ be a function of $X_1, \dots, X_n$, not identically $0$. Suppose that for some $c, r > 0$, $F$ is $c$-Lipschitz and $r$-certifiable, defined as follows:
\begin{itemize}
    \item $F$ is $c$-Lipschitz if changing the outcome of any single $X_i$ changes $F$ by at most $c$.
    
    \item $F$ is $r$-certifiable if, for any $s$, if $F(X_1,\ldots,X_n) \geq s$, then there is a certifying set of at most $rs$ $X_i$'s whose outcomes serve as a witness that $F \geq s$, that is, $F\geq s$ no matter the outcome of the other $X_j$ not in the certifying set.
\end{itemize}
Then, for any $0 \leq t \leq \expect{F}$,
\[
\prob{|F - \expect{F}| > t + 60c\sqrt{r\expect{F}}} \leq 4 \exp\left(-\frac{t^2}{8c^2r\expect{F}}\right).
\]
\end{thm}

The proof of the lemma proceeds by bounding the expected number of exposed balls, then using Talagrand's inequality to achieve a concentration bound.

In what follows, we refer to the balls which are present at the end as $a_1,\ldots,a_k$ and we refer to the remaining balls in the universe as $a_{k + 1}, \ldots, a_\ell$. We denote by $\alpha_i$ the bin choice for $a_i$. For $i \in [k]$, we define $t_i$ to be the last time at which $a_i$ is inserted, we define $X_i$ to be the random variable indicating if $a_i$ is an exposed ball at the end of the game, and we define $X=\sum_{i = 1}^k X_i$ to be the total number of exposed balls.

\begin{clm}\label{thm:iceberg-expectation}
The expected number of exposed balls satisfies $\expect{X} = O(m e^{-\tau^2/(3h)})$.
\end{clm}
\begin{proof}
Recall that $X = \sum_i X_i$ where $X_i$ indicates whether $a_i$ is exposed. By linearity of expectation, it suffices to show that $\expect{X_i} = O(e^{-\tau^2/(3h)})$ for each $i \in [k]$. 

Fix $i \in [k]$. Consider the final time $t_i$ at which ball $a_i$ is inserted. The ball $a_i$ is exposed if and only if the number of balls in bin $\alpha_i$ is at least $h + \tau$. 
If we set $Y$ to be the number of balls in bin $\alpha_i$, and we set $\varepsilon = \tau / h$, then we can bound the probability of $Y \ge h + \tau$ using a Chernoff bound:
\[
    \prob{Y \geq h + \tau} = \prob{Y \geq (1 + \varepsilon)h} \leq e^{-\varepsilon^2h/3} = e^{-\tau^2/(3h)}.
\]
Thus $\prob{X_i} =  e^{-\tau^2/(3h)}$. 
\end{proof}
\begin{clm}\label{thm:iceberg-talagrand}
The random variable $X$ is $(h + \tau + 1)$-Lipschitz and $(h + \tau + 1)$-certifiable as a function of $\{\alpha_i\}_{i = 1}^\ell$.
\end{clm}

\begin{proof}
Changing the value of a single $\alpha_i$ to $\alpha_i'$ can only affect the number of exposed balls in bin $\alpha_i$ (which may decrease) and in bin $\alpha_i'$ (which may increase). The number of \textit{unexposed} balls in a bin is deterministically at most $h + \tau$. This means that moving ball $a_i$ out of bin $\alpha_i$ can increase the number of unexposed balls in the bin by at most $h + \tau$, and thus can decrease the number of exposed balls by at most $h + \tau + 1$ (where the $+ 1$ accounts for the removal of $a_i$ itself). Similarly, moving ball $a_i$ into bin $\alpha_i'$ can decrease the number of unexposed balls in the bin by at most $h + \tau$, and thus can increase the number of exposed balls by at most $h + \tau + 1$. This establishes that $X$ is $(h + \tau + 1)$-Lipschitz.

To certify that $X \geq s$, let $J$ with $|J| = s$ be a set of values $j \in [k]$ such that $a_j$ is exposed at the end of the game. For each $j \in J$, let $R_j$ be a selection of $h + \tau$ balls $i$ such that ball $a_i$ was present at the last time $t_j$ that $a_j$ was inserted and such that $\alpha_i = \alpha_j$. The set of random variables $\{\alpha_i \mid i \in R_j\} \cup \{\alpha_j\}$ acts as a certificate that $a_j$ is exposed. Thus the set
$$\bigcup_{j \in J} \{\alpha_i \mid i \in R_j\} \cup \{\alpha_j\}$$
acts as a certificate that $X \geq s$. This certificate consists of $s (h + \tau + 1)$ random variables, hence $X$ is $(h + \tau + 1)$-certifiable.
\end{proof}

\begin{proof}[Proof of \Cref{lem:iceberg}]
Set $Q = m\exp{(-\tau^2/(3h))}$. By Claim \ref{thm:iceberg-expectation}, we know that $\expect{X} \le Q$.
By Claim \ref{thm:iceberg-talagrand}, we can apply Talagrand's inequality (Theorem \ref{thm:talagrand}) to $X$ with $c = r = h + \tau + 1 = O(h)$.
Applying Talagrand's inequality with $t= \Theta(c \sqrt{r} Q)$, and using $Q$ as an upper bound on $\E[X]$, we can deduce that
$$X = O(c \sqrt{r} Q)$$
with probability at least
$$1 - \exp(-\Omega(Q)).$$
It follows that $X \le \poly(h) \cdot O(ne^{- \tau^2/ (3h)})$ with probability $1 - \exp(-\Omega(me^{-\tau^2/(3h)}))$.
\end{proof}

%% file: ack.tex

\section*{Acknowledgments}

This research was supported in part by NSF grants CSR-1938180, CCF-2106999, CCF-2118620, CCF-2118832, CCF-2106827, CCF-1725543, CSR-1763680, CCF-1716252 and CNS-1938709, as well as an NSF GRFP fellowship and a Fannie and John Hertz Fellowship.

This research was also partially sponsored by the United States Air Force Research Laboratory and the United States Air Force Artificial Intelligence Accelerator and was accomplished under Cooperative Agreement Number FA8750-19-2-1000. The views and conclusions contained in this document are those of the authors and should not be interpreted as representing the official policies, either expressed or implied, of the United States Air Force or the U.S. Government. The U.S. Government is authorized to reproduce and distribute reprints for Government purposes notwithstanding any copyright notation herein.